\documentclass[10pt,a4paper]{article}

\usepackage{amsmath,amssymb,amsthm,mathrsfs}
\usepackage{stmaryrd}  
\usepackage{graphicx}
\usepackage{xcolor}
\usepackage{fullpage}
\usepackage{bm}

\newlength{\defbaselineskip}
\newcommand{\setlinespacing}[1]%
{\setlength{\baselineskip}{#1 \defbaselineskip}}

\theoremstyle{plain}
\newtheorem{theorem}{Theorem}[section]
\newtheorem{lemma}[theorem]{Lemma}
\newtheorem{proposition}[theorem]{Proposition}
\newtheorem{corollary}[theorem]{Corollary}

\theoremstyle{definition}

\theoremstyle{remark}
\newtheorem{remark}[theorem]{Remark}

\numberwithin{equation}{section}

\DeclareMathOperator*{\esssup}{ess\,sup}
\DeclareMathOperator*{\essinf}{ess\,inf}

\allowdisplaybreaks

\begin{document}

\title{
Mean Field Portfolio Games\footnote{We thank the Co-Editor, anonymous Associate Editor, and anonymous referee for many valuable comments and suggestions, which have significantly improved the quality of the paper.}}
\author{Guanxing Fu\footnote{Department of Applied Mathematics, The Hong Kong Polytechnic University, Hung Hom, Kowloon, Hong Kong. Email: guanxing.fu@polyu.edu.hk. G. Fu's research is supported by the Start-up Fund P0035348 from The Hong Kong Polytechnic University.} \quad and \quad Chao Zhou\footnote{Department of Mathematics and Risk Management Institute, National University of Singapore. Email: matzc@nus.edu.sg. C. Zhou’s research is supported by Singapore MOE (Ministry of Educations) AcRF Grants R-146- 000-271-112 and R-146-000-284-114, as well as NSFC Grant No. 11871364.}}

\maketitle

\begin{abstract}
	We study mean field portfolio games with random market parameters, where each player is concerned with not only her own wealth but also  relative performance to her competitors. We use the martingale optimality principle approach to characterize the unique Nash equilibrium in terms of a mean field FBSDE with quadratic growth, which is solvable under a weak interaction assumption. Motivated by the weak interaction assumption, we establish an asymptotic expansion result in powers of the competition parameter. When the market parameters do not depend on the Brownian paths, we obtain the Nash equilibrium in closed form.
\end{abstract}


\section{Introduction}
Mean field games (MFGs) are a powerful tool to study large population games, where each player has negligible influence on the outcome of the game. Introduced independently  by Huang et al.  \cite{HMC-2006} and Lasry and Lions \cite{LL-2007}, MFGs have received considerable attention in the probability and financial mathematics literature. In this paper, we study a classs of mean field portfolio games with random market parameters by the martingale optimality principle (MOP) approach. 

Assume that there are $N$ risky assets in the market, with price dynamics of asset $i\in\{1,\cdots,N\}$ following
\begin{equation}\label{price-i-intro}
	dS^i_t= S^i_t\Big( h^i_t\,dt+\sigma^i_t\,W^i_t+ \sigma^{i0}_t\,dW^0_t		\Big),
\end{equation}
where the return rate  $h^i$ and the volatility $(\sigma^i,\sigma^{i0})$ are assumed to be bounded progressively measurable stochastic processes; $W^i$ is a Brownian motion describing the idiosyncratic noise to the asset $i$; and $W^0$ is a Brownian motion that is independent of $W^i$, describing  common noise to all risky assets. The interest rate of the risk-free asset is assumed to be zero for simplicity. Let $X^i$ be the wealth process of player $i$, who trades asset $i$, and $\overline X^{-i}$ be the  performance index of player $i$.   Each player solves a utility maximization problem and she is concerned with not only her own wealth $X^i$, but also the ``difference'' between her wealth and the performance index. 
 
 We further assume that the risk preference of players is characterized by power utility functions, i.e., player $i\in\{1,2,\cdots,N\}$ chooses the fraction of her wealth invested in the risky asset $i$ to maximize the objective function:
 \begin{equation}\label{PU-intro}
 	\max_{\pi^i}\quad  \mathbb E\left[	\frac{1}{\gamma^i}\left(X^i_T (  \overline X^{-i}_T )^{-\theta^i}	\right)^{\gamma^i}	\right],
 \end{equation}
 where the wealth process $X^i$ follows
 \begin{equation}\label{wealth-Pi-power}
 	\setlength{\abovedisplayskip}{3pt}
 	\setlength{\belowdisplayskip}{3pt}
 	dX^i_t=\pi^i_tX^i_t\Big( h^i_t\,dt+\sigma^i_t\,dW^i_t+\sigma^{i0}_t\,dW^0_t		\Big),\quad X^i_0=x^i,
 \end{equation}
 where $\overline X^{-i}= \left(	\Pi_{j\neq i} X^j	\right)^{\frac{1}{N-1}} $ is the geometric average of all players' wealth except for player $i$; $\gamma^i\in (-\infty,1)/\{0\}$ is the degree of risk aversion; and $\theta^i\in[0,1]$ is the relative competition parameter: player $i$ is more concerned with her own terminal wealth $X^i_T$ if $\theta^i$ is closer to $0$ and more concerned with the relative distance $\frac{X^i_T}{\overline X^{-i}_T}$ between her terminal wealth and the performance index if $\theta^i$ is closer to $1$.  The goal is to find a Nash equilibrium (NE) $(\pi^{1,*},\cdots,\pi^{N,*})$ such that $\pi^{i,*}$ is the optimal strategy for player $i$ and no one wants to change her strategy unilaterally.  
 
By MOP, as in Hu et al. \cite{HIM-2005} and Rouge and El Karoui \cite{REK-2000} for single player's utility maximization problems, the unique NE for the $N$-player game \eqref{PU-intro}-\eqref{wealth-Pi-power} can be characterized by a multidimensional FBSDE with quadratic growth; see Section \ref{sec:model-N-player}. Although such FBSDE is solvable, the equations are tedious. The analysis can be significantly simplified by studying the corresponding MFG:
\begin{equation}\label{model-MFG-power}
	\left\{ \begin{split}
		1.&~\textrm{Fix }\mu\textrm{ in some suitable space};\\
		2.&~\textrm{Solve the optimization problem: }\\
		&~\mathbb E\left[    \frac{1}{\gamma} (X_T\mu^{-\theta}_T)^\gamma 		\right]\rightarrow\max \textrm{ over }\pi\\
		&~\textrm{such that } dX_t=\pi_tX_t(h_t\,dt+\sigma_t\,dW_t+\sigma^0_t\,dW^0_t),~X_0=x_{};\\
		3.&~\textrm{Search for the fixed point }\mu_t=\exp\left(\mathbb E[\log X^*_t|\mathcal F^0_t]\right),~t\in[0,T],\\
		&~X^*\textrm{ is the optimal wealth from }2 \textrm{ and }\mathcal F^0\textrm{ is the filtration generated by }W^0.
	\end{split}\right.
\end{equation}
In the MFG \eqref{model-MFG-power}, by the approach introduced in \cite{HMC-2006,LL-2007}, it is only necessary to consider a representative player's utility maximization problem with $\mu$ fixed, which in turn should be consistent with the aggregation of  optimal wealth. 

In the mathematical finance literature, the first result on portfolio games with relative performance concerns including $N$-player games and MFGs was obtained by Espinosa and Touzi \cite{ET-2015}: in the context of a complete market, the unique NE was established for general utility functions; in the  context of an incomplete market where each player has a heterogeneous portfolio constraint, by assuming that the drift and volatility of the log price are deterministic, using a BSDE approach, \cite{ET-2015}  obtained a unique NE for exponential utility functions, which was called \textit{deterministic Nash equilibrium} in that paper. Moreover, the convergence from $N$-player games to MFGs was also studied in \cite{ET-2015}. Later, Frei and dos Reis \cite{FR-2011} studied similar portfolio games to \cite{ET-2015} from a different perspective: they constructed counterexamples where no NE exists. In contrast to \cite{ET-2015,FR-2011}, where all players trade common stocks, Lacker and Zariphopoulou \cite{LZ-2019} investigated $N$-player and mean field portfolio games where  the stock price follows \eqref{price-i-intro}, but with constant market parameters. Using a PDE approach, \cite{LZ-2019} established a \textit{constant equilibrium} that was proven to be unique among all constant ones. In addition, portfolio games with mean field interaction have been examined in \cite{RP-2021,RP-2021b, HZ-2021,LS-2020}, where in \cite{RP-2021,RP-2021b} dos Reis and Platonov studied $N$-player games and MFGs with forward utilities; and in \cite{LS-2020}, Lacker and Soret extended the CRRA model in \cite{LZ-2019} to include consumption, by using PDE approaches. In \cite{HZ-2021}, Hu and Zariphopoulou studied   portfolio games in an It\^o-diffusion environment.

%


In this paper, we study portfolio games with power utility functions. The games with exponential utility functions and log utility functions can be studied in the same manner; refer to Remark \ref{remark:exponential} and Remark \ref{remark:log}. 
Our paper makes two main contributions. The first one is the wellposedness result of NE for the portfolio game \eqref{model-MFG-power}. We first establish a one-to-one correspondence between the NE of \eqref{model-MFG-power} with random market parameters and some mean field FBSDE with quadratic growth. 
Such correspondence result is key to prove the uniqueness of the NE result. We then solve the FBSDE under a weak interaction assumption, i.e., the competition parameter $\theta$ is assumed to be small. Such assumption is widely used in the game theory and financial mathematics literature; see \cite{CS-2015,CS-2017,FGHP-2018,FH-2018,HMP-2019,Horst-2005}, among others. In order to achieve this, our idea is to first consider the difference between the FBSDE with the benchmark one when $\theta=0$ to cancel out non-homogenous terms, and then transform the resulting FBSDE into a mean field quadratic BSDE. It is worth noting that although a transformation argument from FBSDE to BSDE was also used in \cite{ET-2015}, there is an essential difference from our transformation: the terminal condition of our resulting BSDE is bounded, which makes it convenient to apply the theory of quadratic BSDE. To solve the BSDE, our idea is to decompose the driver into two parts. Specifically, one part does not depend on $\theta$ and the other part depends on $\theta$, and all mean field terms belong to the second part so that they can be controlled by $\theta$.
To the best of our knowledge, both our existence and uniqueness results are new in the literature. In particular, the contribution of our uniqueness result is two-fold. On the one hand, our uniqueness result partially generalizes the uniqueness result in \cite{ET-2015} in the sense that our utility functions can be beyond exponential ones, our stock price is driven by both idiosyncratic and common noise, and both the market parameters and the admissible strategies can be random; however, there is no trading constraint in our paper. On the other hand, when the market parameters are independent of the Brownian paths, we construct a unique NE  in $L^\infty$ in closed form, which is beyond constant strategies.  This result completely generalizes the result in \cite{LZ-2019}, where only constant NE was studied. Such generalization strongly relies on our FBSDE approach, which shares a similar idea to \cite{CD-2013,ET-2015,FR-2011} and is more powerful than the PDE approach in \cite{LZ-2019}.

Our second contribution is an asymptotic expansion result. Motivated by the weak interaction assumption, we provide an approximation in any order of the value function and the optimal investment for the model with competition in terms of the solutions to the benchmark model without competition when the investor is only concerned with her own wealth. These results enable us to obtain the value function and the optimal investment based only on the benchmark model in the case of a small competition parameter. In order to obtain the asymptotic expansion results, our idea is to start with the FBSDE charaterization of the NE, and establish the expansion of the solution to the FBSDE by studying the iterative system of FBSDEs with coefficients in the BMO space. Our analysis relies on the application of energy inequality and reverse H\"older's inequality for BMO martingales. Asymptotic expansion results in stochastic optimization and game setting were also studied in \cite{Horst2022} and \cite{CS-2015,CS-2017} by PDE analysis. In \cite{Horst2022}, Horst et al. investigated a single-player optimal liquidation problem under ambiguity with respect to price impact parameters. They established a first-order approximation result of the robust model for small uncertainty factors, while our approximation is in any order. In \cite{CS-2015,CS-2017}, Chan and Sircar analyzed continuous time Bertrand and Cournot competitions as MFGs. Instead of solving the forward-backward PDE characterizing the NE, they established a formal asymptotic expansion result in powers of the competition parameter; no rigorous proof of the asymptotic result was provided there.
It is worth noting that the market parameters in \cite{CS-2015,CS-2017,Horst2022} are deterministic, while ours are random.




The remainder of this paper is organized as follows: \\
$\bullet$ After the introduction of notation, in Section \ref{sec: characterization}, we establish an equivalent relationship between the existence of NE of the MFG \eqref{model-MFG-power} and the solvability of some mean field FBSDE. \\
$\bullet$ In Section \ref{sec:wellposedness}, we study the MFG \eqref{model-MFG-power} in detail; in particular, Section \ref{sec:wellposedness} addresses the wellposedness of the MFG with general market parameters by solving the FBSDE introduced in Section \ref{sec: characterization}. Moreover, when the market parameters do not depend on the Brownian paths, we find the NE in closed form.\\
$\bullet$ In Section \ref{sec:asymptotic}, the asymptotic result is established. Specifically, the logarithm of value function and the optimal investment are expanded into any order in powers of $\theta$. \\
$\bullet$
In Section \ref{sec:model-N-player}, we comment on the result of $N$-player games. 
\paragraph{Notation} In the probability space $(\Omega,\mathbb P,\mathbb K=(\mathcal K_t)_{0\leq t\leq T})$, a two-dimensional Brownian motion $\overline W=(W,W^0)^\top$ is defined, where $W$ is the idiosyncratic noise for the representative player, and $W^0$ is the common noise for all players. In addition, $\mathbb G=\{ \mathcal G_t, t\in[0,T]\}$ is assumed to be the augmented natural filtration of $\overline W$. The augmented natural filtration of $W^0$ is denoted by $\mathbb F^0=\{\mathcal F^0_t,t\in[0,T]\}$. Let $\mathcal A$ be a $\sigma$-algebra that is independent of $\mathbb G$. Let $\mathbb F=\{\mathcal F_t, t\in[0,T]	\}$ be the $
\sigma$-algebra generated by $\mathcal A$ and $\mathbb G$.

For a random variable $\xi$, let $\|\xi\|$ be the essential supremum of its absolute value $|\xi|$, and let $\underline\xi$ be its essential infimum. For a sub $\sigma$-algebra $\mathcal H$ of $\mathbb F$, let $\textrm{Prog}(\Omega\times[0,T];\mathcal H)$ be the space of all stochastic processes that are $\mathcal H$-progressively measurable. For each $\eta\in \textrm{Prog}(\Omega\times[0,T];\mathcal H)$, define $\|\eta\|_\infty=\esssup_{\omega\in\Omega,t\in[0,T]} |\eta_t(\omega)| $. Let $L^\infty_{\mathcal H}$ be the space of all essentially bounded stochastic processes, i.e.:
\[
L^\infty_{\mathcal H}=\{	\eta\in\textrm{Prog}(\Omega\times[0,T];\mathcal H): \|\eta\|_\infty<\infty	\}.
\] 
Let $S^\infty_{\mathcal H}$ be the subspace of $L^\infty_{\mathcal H}$, where the trajectories of all processes are continuous. Furthermore, for a probability measure $\mathbb Q$ and for $p>1$, define
\[
S^p_{\mathbb Q,\mathcal H}=\left\{\eta\in \textrm{Prog}(\Omega\times[0,T];\mathcal H): \eta\textrm{ has continuous trajectory and }\|\eta\|_{S^p_{\mathbb Q,\mathcal H}}:=\left(\mathbb E^{\mathbb Q}\left[\sup_{0\leq t\leq T}|\eta_t|^p\right]\right)^{1/p}<\infty			\right\}
\]
and
\[
M^p_{\mathbb Q,\mathcal H}=\left\{	\eta\in\textrm{Prog}(\Omega\times[0,T];\mathcal H): \|\eta\|_{M^p_{\mathbb Q,\mathcal H}}:= \left(\mathbb E^{\mathbb Q}\left[	\left(\int_0^T|\eta_t|^2\,dt\right)^{\frac{p}{2}}	\right]\right)^{1/p}	<\infty	\right\}.
\]
Define the BMO space under $\mathbb Q$ as
\[
H^2_{BMO,\mathbb Q,\mathcal H}=\left\{ \eta\in\textrm{Prog}(\Omega\times[0,T];\mathcal H):		\|\eta\|_{BMO,\mathbb Q,\mathcal H}^2:= \sup_{\tau:\mathcal H-\textrm{stopping time}}\left\|\mathbb E^{\mathbb Q}\left[ \left.	\int_\tau^T |\eta_t|^2\,dt	\right|\mathcal F_\tau	\right]\right\|_\infty	<\infty	\right\}.
\]
In particular, if $\mathbb Q=\mathbb P$, which is the physical measure, and/or $\mathcal H=\mathbb F$, we drop the dependence on $\mathbb Q$ and/or $\mathcal H$ in the definition of the above spaces.
\paragraph{Assumption 1.} 
The initial wealth $x$, risk aversion parameter $\gamma$, and competition parameter $\theta$ of the population are assumed to be bounded  $\mathcal A$-measurable random variables. In addition, $x$ is valued in $(0,\infty)$,  $\gamma$ is valued in $(-\infty,1)/\{ 0\}$, and $\theta$ is valued in $[0,	1]$.

Assume the return rate $h\in L^\infty$ and the volatility $\widetilde\sigma:=(\sigma,\sigma^0)^\top\in L^\infty\times L^\infty$. Furthermore, $|\gamma|$, $|\sigma|+|\sigma^0|$ are bounded away from $0$, i.e., $|\gamma|\geq \underline\gamma>0$ a.s. and $\essinf_{\omega\in\Omega}\inf_{t\in[0,T]}\widetilde\sigma^\top\widetilde\sigma>0$.

\paragraph{Space of Admissible Strategies.} We assume that the space of admissible strategies for the representative player is $H^2_{BMO}$.

\textbf{Definition of NE.} We say that the pair $(\mu^*,\pi^*)$ is an NE of \eqref{model-MFG-power}, if $\pi^*\in H^2_{BMO}$, $\mu^*_t=\exp\left(\mathbb E[\log X^*_t|\mathcal F^0_t]\right)$, for $t\in[0,T]$, and $\mathbb E\left[    \frac{1}{\gamma} \Big(X^{\pi^*}_T(\mu^*_T)^{-\theta}\Big)^\gamma 		\right]\geq \mathbb E\left[    \frac{1}{\gamma} \Big(X^\pi_T(\mu^*_T)^{-\theta}\Big)^\gamma 		\right]$ for each admissible strategy $\pi$. Specifically, $(\mu^*_t)_{0\leq t\leq T}$ is called a solution to \eqref{model-MFG-power}.

\begin{remark}
	(1) In \textbf{Assumption 1}, the assumption that ($x$, $\gamma$, $\theta$) is  $\mathcal A$-measurable is consistent with the formulation in \cite[Remark 5.10]{ET-2015} and the \textit{random type} introduced in \cite{LZ-2019};
	
	(2) The space of admissible strategies is consistent with \cite{FR-2011} and \cite{MPZ-2015};
	
	(3) The NE in \textbf{Definition of NE} is of open-loop type. We make a remark on the closed-loop NE in Remark \ref{rmk:closed-lopp}.    
\end{remark}

\section{MFGs and Mean Field FBSDEs Are Equivalent}\label{sec: characterization}
In this section, we prove that the solvability of some mean field FBSDE is sufficient and necessary for the solvability of the MFG \eqref{model-MFG-power}. This equivalent result is key to establish the uniqueness result of NE. The sufficient part is proven by MOP in \cite{HIM-2005}, and the necessary part is proven by the dynamic programming principle in \cite[Theorem 4.7]{ET-2015} and \cite[Lemma 3.2]{FR-2011}, where the $N$-player game with exponential utility functions and trading constraint, but without idiosyncratic noise, was studied. In the next  proposition, we adapt the argument to our MFG  \eqref{model-MFG-power} with power utility functions. 

\begin{proposition}\label{lem:NE-BSDE-power}
$\bm{(1)}$ If an NE $(\mu^*,\pi^*)$ of the MFG \eqref{model-MFG-power} exists, with $(\widehat\mu^*,\pi^*)\in S^2_{\mathbb F^0}\times H^2_{BMO} $ and 
	\begin{equation}\label{necessary:reverse-power}
		\mathbb E\left[	\frac{1}{\gamma} e^{\gamma (\widehat X^{\pi^*}_T-\theta\widehat\mu^*_T)} \Big|\mathcal F_\cdot	\right] \textrm{ satisfying }R_p\textrm{ for some }p>1, \footnote{Refer to Appendix \ref{app:reverse} for the definition of the condition $R_p$.}
	\end{equation}	
where $\widehat X^{\pi^*}=\log X^{\pi^*}$ is the log-wealth and $\widehat\mu^*=\log\mu^*$, 
	then the following mean field FBSDE admits a solution, such that $(Z,Z^0)\in H^2_{BMO}\times H^2_{BMO}$,
	\begin{equation}\label{FBSDE-power-2}
		\left\{\begin{split}
			d\widehat X_t=&~	\frac{	h_t+\sigma_tZ_t+\sigma^0_tZ^0_t	}{(1-\gamma)( \sigma^2_t+(\sigma^0_t)^2 )}\left\{ \left(	h_t-\frac{h_t+\sigma_tZ_t+\sigma^0_tZ^0_t}{2(1-\gamma)}\right)\,dt+\sigma_t\,dW_t+\sigma^0_t\,dW^0_t		\right\}, 		\\
			-dY_t=&~\left(\frac{Z^2_t+(Z^0_t)^2}{2}+\frac{\gamma}{2(1-\gamma)}\frac{ (h_t+\sigma_t Z_t+\sigma^0_tZ^0_t)^2    }{\sigma^2_t+(\sigma^0_t)^2}\right)\,dt-Z_t\,dW_t-Z_t^0\,dW^0_t,\\
			\widehat X_0=&~\log(x_{}),~Y_T=-\gamma\theta\mathbb E[\widehat X_T|\mathcal F^0_T].
		\end{split}\right.
	\end{equation}
	
$\bm{(2)}$ If the FBSDE \eqref{FBSDE-power-2} admits a solution, such that $(Z,Z^0)\in H^2_{BMO}\times H^2_{BMO}$, then the MFG \eqref{model-MFG-power} admits an NE $(\mu^*,\pi^*)$, such that $(\widehat\mu^*,\pi^*)\in S^2_{\mathbb F^0}\times H^2_{BMO}$ and \eqref{necessary:reverse-power} holds. 
	
	The relationship is given by $\pi^*=\frac{h+\sigma Z+\sigma^0Z^0}{(1-\gamma)(  \sigma^2+(\sigma^0)^2  )}$.
\end{proposition}
\begin{proof}
$\bm{(1)}$ Let $(\mu^*,\pi^*)$ be an NE of \eqref{model-MFG-power} such that \eqref{necessary:reverse-power} holds. Define
\[
M^\pi_t=e^{\gamma\widehat X^\pi_t} \esssup_{\kappa\in H^2_{BMO}} \mathbb E\left[ \frac{1}{\gamma}e^{\gamma(\widehat X_T^\kappa-\widehat X^\kappa_t-\theta\widehat\mu^*_T)}	\Big|\mathcal F_t		\right],
\]
where $\widehat X^\kappa$ is the log-wealth associated with the strategy $\kappa$.
Following the argument in \cite[Theorem 4.7]{ET-2015} and \cite[Lemma 3.2]{FR-2011}, $M^\pi$ has a continuous version which is a supermartingale for all $\pi$ and a martingale for $\pi^*$, and there exists a $\breve Z\in H^2_{BMO}\times H^2_{BMO}$, such that
\[
M^{\pi^*}_t=M^{\pi^*}_0e^{\int_0^t\breve Z^\top_s\,d\overline W_s-\frac{1}{2}\int_0^t\breve Z^\top_s\breve Z_s\,ds}.
\]
Straightforward calculation implies that
\begin{equation}
	\begin{split}
		M^{\pi}_t=&~e^{-\gamma( \widehat X^{\pi^*}_t-\widehat X^\pi_t )}M^{\pi^*}_t\\
		=&~M^{\pi^*}_0\mathcal E\Big(\int_0^t ( -\gamma\pi^*_s\widetilde\sigma^\top_s+\gamma\pi_s\widetilde\sigma^\top_s+\breve Z^\top_s  )\,d\overline W_s		\Big)\exp\Big(\int_0^t\widetilde f_s\,ds	\Big),
	\end{split}
\end{equation}
where
\[
\widetilde f=-\gamma\pi^*h+\frac{\gamma}{2}(\pi^*)^2\widetilde\sigma^\top\widetilde\sigma +\gamma\pi h-\frac{\gamma}{2}\pi^2\widetilde\sigma^\top\widetilde\sigma+\frac{1}{2}\Big(	\gamma\pi^*\widetilde\sigma^\top-
\gamma\pi\widetilde\sigma^\top \Big)\Big(	\gamma\pi^*\widetilde\sigma-
\gamma\pi\widetilde\sigma \Big) +(-\gamma\pi^*\widetilde\sigma^\top+\gamma\pi\widetilde\sigma^\top )\breve Z.
\]
Let $\mathring Z=\breve Z-\gamma\pi^*\widetilde\sigma$. Then $\widetilde f$ can be rewritten as
\[
\widetilde f=\frac{\gamma}{2(1-\gamma)\widetilde\sigma^\top\widetilde\sigma}\Big|(1-\gamma)\widetilde\sigma^\top\widetilde\sigma\pi^*-h-\widetilde\sigma^\top\mathring Z		\Big|^2-  \frac{\gamma}{2(1-\gamma)\widetilde\sigma^\top\widetilde\sigma}\Big|(1-\gamma)\widetilde\sigma^\top\widetilde\sigma\pi-h-\widetilde\sigma^\top\mathring Z		\Big|^2.
\]
Since $M^\pi$ is a supermartingale, $\exp\Big(\int_0^\cdot f_s\,ds	\Big)$ is nonincreasing if $\gamma>0$ and nondecreasing if $\gamma<0$. As a result, $\frac{f}{\gamma}$ is nonpositive. Thus,
$\pi^*=\frac{h+\widetilde\sigma^\top\mathring Z}{(1-\gamma)\widetilde\sigma^\top\widetilde\sigma}$. Define $Y=\log\Big(	M^{\pi^*}\exp(-\gamma\widehat X^{\pi^*})\Big)$. Then $(\log X^{\pi^*},Y,\mathring Z)$ satisfies \eqref{FBSDE-power-2}.

$\bm{(2)}$ For each strategy $\pi\in H^2_{BMO}$, define $R^{\pi}_t=\frac{1}{\gamma}X^{\gamma}_t e^{Y_t}$, where $dY_t=f_t(Z_t,Z^0_t)\,dt+Z_t\,dW_t+Z^0_t\,dW^0_t$, $Y_T=-\gamma\theta\mu^*_T$ with $f$ to be determined, such that
\begin{equation}\label{3points-power}
	\begin{split}
		\bullet&~ R^{\pi}_0\textrm{ is independent of }\pi;\\
		\bullet&~ R^\pi\textrm{ is a supermartingale for all }\pi\textrm{ and a martingale for some }\pi^*;\\
		\bullet&~ R^\pi_T=\frac{1}{\gamma}(X_T (\mu_T^*)^{-\theta} )^\gamma.
	\end{split}
\end{equation}
The three points in \eqref{3points-power} indicate that $\mathbb E[R^{\pi^*}_T]=\mathbb E[R^{\pi^*}_0]=\mathbb E[R^\pi_0]\geq \mathbb E[R^\pi_T]$ for all $\pi$.

Note that
\begin{equation*}
	\begin{split}
		R^\pi_t
		=&~\frac{1}{\gamma}x_{}^\gamma\exp(Y_0)\exp\left(	\int_0^t	\left(  \gamma\pi_sh_s-\frac{\gamma\pi^2_s}{2}(\sigma^2_s+(\sigma^0_s)^2	 )	+f_s(Z_s,Z^0_s)	 +\frac{1}{2}( \gamma\pi_s\sigma_s+Z_s )^2 +\frac{1}{2}(  \gamma\pi_s\sigma^0_s+Z^0_s )^2 \right)	\,ds 	\right)\\
		&~\times \mathcal E\left( \int_0^t(\gamma\pi_s\sigma_s+Z_s)\,dW_s+\int_0^t(\gamma\pi_s\sigma^0_s+Z^0_s)\,dW^0_s		\right)\\
		:=&~  \frac{1}{\gamma}x_{}^\gamma\exp(Y_0)\exp\left( \int_0^t\widetilde f_s(\pi_s,Z_s,Z^0_s)\,ds	\right)\mathcal E\left( \int_0^t(\gamma\pi_s\sigma_s+Z_s)\,dW_s+\int_0^t(\gamma\pi_s\sigma^0_s+Z^0_s)\,dW^0_s		\right).
	\end{split}
\end{equation*}
Since $\pi\in H^2_{BMO}$, $\mathcal E\left( \int_0^t(\gamma\pi_s\sigma_s+Z_s)\,dW_s+\int_0^t(\gamma\pi_s\sigma^0_s+Z^0_s)\,dW^0_s		\right)$ is a martingale. In order to make $R^\pi$ satisfy the second point of \eqref{3points-power}, we choose $f$, such that $\widetilde f(\pi,Z,Z^0)$ is nonpositive for all $\pi$ and zero for some $\pi^*$. By rearranging terms we have the following equation
\begin{equation*}
	\begin{split}
		\widetilde f(\pi,Z,Z^0)=  -\frac{\gamma-\gamma^2}{2}(\sigma^2+(\sigma^0)^2		) \pi^2	+(	\gamma h+\sigma\gamma Z+\gamma\sigma^0 Z^0	)\pi+\frac{Z^2+(Z^0)^2}{2}+f(Z,Z^0).
	\end{split}
\end{equation*}
By choosing
\[
\pi^*=\frac{h+\sigma Z+\sigma^0Z^0}{(1-\gamma)(  \sigma^2+(\sigma^0)^2  )}
\]
and
\[
f(Z,Z^0)=-\frac{Z^2+(Z^0)^2}{2}-\frac{\gamma}{2(1-\gamma)}\frac{ (h+\sigma Z+\sigma^0Z^0)^2    }{\sigma^2+(\sigma^0)^2},
\]
it holds that $\widetilde f$ is nonpositive for all $\pi$, and $\widetilde f(\pi^*,Z,Z^0)=0$. 
Thus, by \eqref{3points-power}, an NE of \eqref{model-MFG-power} exists if the following FBSDE admits a solution with $(Z,Z^0)\in H^2_{BMO}\times H^2_{BMO}$
\begin{equation}\label{FBSDE-power-1}
	\left\{\begin{split}
		dX_t=&~	\frac{	h_t+\sigma_tZ_t+\sigma^0_tZ^0_t	}{(1-\gamma)( \sigma^2_t+(\sigma^0_t)^2 )}X_t( h_t\,dt+\sigma_t\,dW_t+\sigma^0_t\,dW^0_t		)		\\
		-dY_t=&~\left(\frac{Z^2_t+(Z^0_t)^2}{2}+\frac{\gamma}{2(1-\gamma)}\frac{ (h_t+\sigma_t Z_t+\sigma^0_tZ^0_t)^2    }{\sigma^2_t+(\sigma^0_t)^2}\right)\,dt-Z_t\,dW_t-Z_t^0\,dW^0_t,\\
		X_0=&~x,~Y_T=-\gamma\theta \mathbb E[\log X_T|\mathcal F^0_T]
	\end{split}\right.
\end{equation}
Let $\widehat X= \log X$. \eqref{FBSDE-power-1} is equivalent to \eqref{FBSDE-power-2}.	
\end{proof}

\begin{remark}
	The proof of $\bm{(2)}$ in Proposition \ref{lem:NE-BSDE-power} relies on MOP in \cite{HIM-2005}. The essential difference between our proof and \cite{HIM-2005} is  the choice of strategies; we consider the fraction of the
	wealth invested in stock $\pi$ as our strategy, while in \cite{HIM-2005} Hu et al.  considered the scaled one $\widetilde\pi:=\pi (\sigma,\sigma^0)$ as a strategy. We claim that the choice in \cite{HIM-2005} is not appropriate in the game-theoretic version of utility maximization problems. The reason is that $\sigma$ and $\sigma^0$ do not have symmetric status in the game; the former is the volatility of the idiosyncratic noise, while the latter is the volatility of the common noise.  To illustrate the difference resulting from the choice of strategies, we assume that all coefficients are  $\mathcal A$-measurable random variables. Following the argument in \cite[Section 3]{HIM-2005}, the optimal strategy is given by 
	\[
			\widetilde\pi^* =\frac{1}{1-\gamma} (\widetilde {\bm Z}+\Theta), 
	\] 
	where $\Theta=\left(\frac{\sigma h}{\sigma^2+(\sigma^0)^2},\frac{\sigma^0h}{\sigma^2+(\sigma^0)^2}		\right)$, and $\widetilde {\bm Z}=(Z,Z^0)$ together with some $(\widetilde {\bm X},\widetilde {\bm Y})$ satisfies the   FBSDE 
	\begin{equation}\label{FBSDE:HIM}
		\left\{\begin{split}
			d\widetilde {\bm X}_t=&~\widetilde \pi^*_t\left(\Theta^\top -\frac{1}{2}(\widetilde\pi^*_t)^\top		\right)\,dt+\widetilde\pi^*_t\,d\overline W_t,\\
			-d\widetilde {\bm Y}_t= &~ \frac{\gamma| \widetilde {\bm Z}_t+\Theta |^2}{2(1-\gamma)} + \frac{|\widetilde {\bm Z}_t|^2}{2} \,dt-\widetilde {\bm Z}_t\,d\overline W_t,\\
			\widetilde {\bm X}_0=&~\log(x),\quad \widetilde{\bm Y}_T=-\gamma\theta\widehat\mu^*_T,
		\end{split}\right.
	\end{equation}
	with $\widehat\mu^*_T=\mathbb E[\widetilde {\bm X}_T|\mathcal F^0_T]$.  One can verify directly that the $Z$-component of the solution to \eqref{FBSDE:HIM} is
	\[
	\widetilde{\bm Z}=	\left( 0,   -\frac{ \theta\gamma\mathbb E\left[ \frac{\sigma^0 h}{(1-\gamma)( \sigma^2+(\sigma^0)^2 )}			\right] }{  1+\mathbb E\left[  \frac{\theta\gamma}{1-\gamma}		\right]    }				\right). 
	\]
	Thus, the optimal strategy is 
	\begin{equation}\label{pi-HIM}
		\widetilde \pi^*=\frac{1}{1-\gamma}\left(	 \frac{\sigma h}{\sigma^2+(\sigma^0)^2}	,  	~  \frac{\sigma^0h}{\sigma^2+(\sigma^0)^2}-\frac{ \theta\gamma\mathbb E\left[ \frac{\sigma^0 h}{(1-\gamma)( \sigma^2+(\sigma^0)^2 )}			\right] }{  1+\mathbb E\left[  \frac{\theta\gamma}{1-\gamma}		\right]    }		\right):=\pi^*(\sigma,\sigma^0). 
	\end{equation}
Multiplying $(\sigma,\sigma^0)^\top$ on both sides of \eqref{pi-HIM}, we get
	\[
		\pi^*=\frac{1}{(1-\gamma)( \sigma^2+(\sigma^0)^2 )} \left\{	h-\frac{\theta\gamma\sigma^0\mathbb E\left[\frac{\sigma^0h}{(1-\gamma)(\sigma^2+(\sigma^0)^2)}	\right]}{1+\mathbb E\left[\frac{\theta\gamma}{1-\gamma}	\right]}			\right\},
	\]
which is surprisingly different from our Theorem	 \ref{thm:constant-equilibrium-MFG} and Corollary \ref{corollary:LZ}, unless in the following two special cases:      

$\bullet$ $\sigma=0:$ in this case, all players trade a common stock; this case was considered in \cite{ET-2015,FR-2011}.
\newline	$\bullet$ $\theta=0$: in this case, there is no competition; this is the single-player utility maximization problem considered in \cite{HIM-2005}.						

Given the above argument, we provided the detailed proof in Proposition \ref{lem:NE-BSDE-power}$\bm{(2)}$ for readers' convenience.
\end{remark}

\begin{remark}[MFGs with exponential utility functions] 	\label{remark:exponential}
If each player uses an exponential utility criterion, then the MFG becomes
	\begin{equation}\label{model-MFG-exp}
		\left\{ \begin{split}
			1.&~\textrm{Fix }\mu\textrm{ in some suitable space};\\
			2.&~\textrm{Solve the optimization problem: }\\
			&~\mathbb E\left[    - e^{-\alpha(X_T-\theta\mu_T)} 		\right]\rightarrow\max \textrm{ over }\pi\\
			&~\textrm{such that } dX_t=\pi_t(h_t\,dt+\sigma_t\,dW_t+\sigma^0_t\,dW^0_t),~X_0=x_{exp};\\
			3.&~\textrm{Search for the fixed point }\mu_t=\mathbb E[X^*_t|\mathcal F^0_t],~t\in[0,T],\\
			&~X^*\textrm{ is the optimal wealth from }2.
		\end{split}\right.
	\end{equation}
	By the same analysis as in Proposition \ref{lem:NE-BSDE-power}, the existence of NE of the MFG \eqref{model-MFG-exp} is equivalent to the solvability of the following FBSDE
	\begin{equation}\label{MF-FBSDE}
		\left\{\begin{split}
			dX_t=&~\frac{\alpha\sigma_t Z_t+\alpha\sigma^0_tZ^0_t+h_t  }{\alpha( \sigma^2_t+(\sigma^0_t)^2 )}(h_t\,dt+\sigma_t\,dW_t+\sigma^0_t\,dW^0_t),\\
			dY_t=&~\left(\frac{(\alpha\sigma_t Z_t+\alpha\sigma^0_t Z^0_t+h_t)^2}{ 2\alpha( \sigma^2_t+(\sigma^0_t)^2 )  }-\frac{\alpha}{2}( Z_t^2+(Z^0_t)^2 )\right)\,dt+Z_t\,dW_t+Z^0_t\,dW^0_t,\\
			X_0=&~x_{exp},~Y_T=\theta\mathbb E[X_T|\mathcal F^0_T].
		\end{split}\right.
	\end{equation}
	The relationship is given by $\pi^*=\frac{  \alpha\sigma Z+\alpha\sigma Z^0+h			}{\alpha\big( \sigma^2+(\sigma^0)^2\big)}$.  The FBSDE \eqref{MF-FBSDE} can be solved in exactly the same manner as \eqref{FBSDE-power-2}. Consequently, in this paper, we will only consider the game with power utility functions, and the analysis of the exponential case is available upon request. 				
\end{remark}

\begin{remark}[MFGs with log utility functions]\label{remark:log}
	If each player uses a log utility criterion, then the MFG becomes
	\begin{equation}\label{model-MFG-log}
		\left\{ \begin{split}
			1.&~\textrm{Fix }\mu\textrm{ in some suitable space};\\
			2.&~\textrm{Solve the optimization problem: }\\
			&~\mathbb E\left[    \log \big(X_T\mu^{-\theta}_T\big) 		\right]\rightarrow\max \textrm{ over }\pi\\
			&~\textrm{such that } dX_t=\pi_tX_t(h_t\,dt+\sigma_t\,dW_t+\sigma^0_t\,dW^0_t),~X_0=x_{log};\\
			3.&~\textrm{Search for the fixed point }\mu_t=\exp\left(\mathbb E[\log X^*_t|\mathcal F^0_t]\right),~t\in[0,T],\\
			&~X^*\textrm{ is the optimal wealth from }2.
		\end{split}\right.
	\end{equation}
Note that $\arg\max_\pi\mathbb E\left[    \log \big(X^\pi_T\mu^{-\theta}_T\big) 		\right]=\arg\max_\pi\mathbb E[\log X^\pi_T]$. Therefore, the MFG with log utility criterion is decoupled; each player makes her decision by disregarding her competitors. By \cite{HIM-2005}, the NE of \eqref{model-MFG-log} is given by
\begin{equation}\label{NE-log}
	\mu^*=\exp\Big(	\mathbb E[\log X_T|\mathcal F^0_T]\Big),					\qquad \pi^*=\frac{h}{\sigma^2+(\sigma^0)^2},
\end{equation}
where $X$ together with some $(Y,Z)$ is the unique solution to the (trivially solvable) FBSDE 
\begin{equation}\label{FBSDE-log}
	\left\{\begin{split}
		dX_t=&~\frac{h_t}{\sigma^2_t+(\sigma^0_t)^2}X_t(	h_t\,dt+\sigma_t\,dW_t+\sigma^0_t\,dW^0_t	),\\
		-dY_t=&~\frac{h^2_t}{2(\sigma^2_t+(\sigma^0_t)^2)}\,dt+Z_t\,dW_t+Z^0_t\,dW^0_t,\\
		X_0=&~x_{log},~Y_T=-\theta\mathbb E[\log X_T|\mathcal F^0_T].
	\end{split}\right.
\end{equation}
Let $(\mu',\pi')$ be any other NE of \eqref{model-MFG-log}. Given $\mu'$, by MOP in \cite{HIM-2005}, the optimal response is $\frac{h}{\sigma^2+(\sigma^0)^2}$, which is unique since the log utility function is concave. Thus, $\pi'=\frac{h}{\sigma^2+(\sigma^0)^2}$ and $\mu'_t=\exp\Big(\mathbb	E[\log X_t|\mathcal F^0_t]\Big)$, $t\in[0,T]$ and $X$ is the unique solution to \eqref{FBSDE-log}. Therefore, $(\mu^*,\pi^*)=(\mu',\pi')$ and the NE of \eqref{model-MFG-log} is unique. 
\end{remark}

In Section \ref{sec:wellposedness}, we will study the MFG \eqref{model-MFG-power} by examining the FBSDE   \eqref{FBSDE-power-2}. In particular, we will  consider the difference between the FBSDE \eqref{FBSDE-power-2} with the benchmark one when the competition parameter $\theta=0$, and then transform the resulting FBSDE into some BSDE.

\section{Wellposedness of the FBSDE \eqref{FBSDE-power-2} and the MFG   \eqref{model-MFG-power}}\label{sec:wellposedness}

\subsection{The Adjusted Mean Field FBSDE}
By Proposition \ref{lem:NE-BSDE-power}, to solve the MFG  \eqref{model-MFG-power}, it is equivalent to solve the mean field FBSDE  \eqref{FBSDE-power-2}.  
In order to solve \eqref{FBSDE-power-2}, we 
compare \eqref{FBSDE-power-2} with the benchmark FBSDE associated with the single player's utility maximization problem, i.e., the utility game with $\theta=0$. When $\theta=0$, \eqref{FBSDE-power-2} is decoupled into
\begin{equation}\label{MF-FBSDE-0}
	\left\{\begin{split}
		d\widehat X_t=&~	\frac{	h_t+\sigma_tZ_t+\sigma^0_tZ^0_t	}{(1-\gamma)( \sigma^2_t+(\sigma^0_t)^2 )}\left\{ \left(	h_t-\frac{h_t+\sigma_tZ_t+\sigma^0_tZ^0_t}{2(1-\gamma)}\right)\,dt+\sigma_t\,dW_t+\sigma^0_t\,dW^0_t		\right\}, 		\\
	-dY_t=&~\left(\frac{Z^2_t+(Z^0_t)^2}{2}+\frac{\gamma}{2(1-\gamma)}\frac{ (h_t+\sigma_t Z_t+\sigma^0_tZ^0_t)^2    }{\sigma^2_t+(\sigma^0_t)^2}\right)\,dt-Z_t\,dW_t-Z_t^0\,dW^0_t,\\
	\widehat X_0=&~\log(x_{}),~Y_T=0.
	\end{split}\right.
\end{equation}
The solvability of the FBSDE \eqref{MF-FBSDE-0} is summarized in the following proposition.
	\begin{proposition}\label{lem:YZ-o}
	The FBSDE \eqref{MF-FBSDE-0} has a unique solution in $\bigcap_{p>1}S^p\times S^\infty\times H^2_{BMO}\times H^2_{BMO}$. 
\end{proposition}
\begin{proof}
	Theorem 7 in \cite{HIM-2005} implies that there exists a unique $(Y,Z,Z^0)\in S^\infty\times H^2_{BMO}\times H^2_{BMO}$ satisfying the BSDE in \eqref{MF-FBSDE-0}. By the energy inequality (\cite[P.26]{Kazamaki-2006}), it holds that $(Z,Z^{0})\in \bigcap\limits_{p>1} M^p\times\bigcap\limits_{p>1} M^p$, which implies that $\widehat X\in\bigcap\limits_{p>1}S^p$. \\[-10mm] 

\end{proof}
From now on, we
denote the unique solution to \eqref{MF-FBSDE-0} by $(X^{\bar o}, Y^{\bar o}, Z^{\bar o}, Z^{0,\bar o})$.
Let $(\widehat X,Y,Z,Z^0)$ be a solution to \eqref{FBSDE-power-2} and we consider the difference 
\begin{equation}\label{eq:diff_XYZ}
	(\overline X,\overline Y,\overline Z,\overline Z^0):=(\widehat X-X^{\bar o},Y-Y^{\bar o},Z-Z^{\bar o},Z^0-Z^{0,\bar o}),
\end{equation}
which satisfies
\begin{equation}\label{MF-FBSDE-diff}
	\left\{\begin{split}
		d\overline X_t=&~\bigg\{\frac{	 \sigma_t\overline Z_t+ \sigma^0_t\overline Z^0_t	}{(1-\gamma)(\sigma^2_t+(\sigma^0_t)^2	)}	\left(	h_t-	\frac{1}{1-\gamma}(	h_t+\sigma_tZ^{\bar o}_t+\sigma^0_tZ^{0,\bar o}_t	)	\right)	\\
		&~-\frac{	\left(	\sigma_t\overline Z_t+\sigma^0_t\overline Z^0_t	\right)^2	}{2(1-\gamma)^2(	\sigma^2_t+(\sigma^0_t)^2	)} \bigg\}\,dt +\frac{\sigma_t\overline Z_t+\sigma^0_t\overline Z^0_t}{(1-\gamma)(\sigma^2_t+(\sigma^0_t)^2)}  (\sigma_t\,dW_t+\sigma^0_t\,dW^0_t)  \\
		-d\overline Y_t=&~\left\{ \frac{( 2Z^{\bar o}_t+\overline Z_t )\overline Z_t+( 2Z^{0,\bar o}_t+\overline Z^0_t )\overline Z^0_t}{2}  \right.\\
		&~\left. +  \frac{	\gamma(	2h_t+2\sigma_tZ^{\bar o}_t+2\sigma^0_tZ^{0,\bar o}_t+\sigma_t\overline Z_t+\sigma^0_t\overline Z^0_t		)(	\sigma_t\overline Z_t+\sigma^0_t\overline Z^0_t		)		}{2(1-\gamma)(\sigma^2_t+(\sigma^0_t)^2)}  \right\}\,dt\\
		&~-\overline Z_t\,dW_t-\overline Z^0_t\,dW^0_t,\\
		\overline X_0=&~0,~\overline Y_T=-\theta\gamma\mathbb E[\overline X_T|\mathcal F^0_T]-\theta\gamma\mathbb E[X^{\bar o}_T|\mathcal F^0_T].
	\end{split}\right.
\end{equation}

\begin{remark}
	There are two reasons why we consider the difference instead of the original FBSDE. First, we want to solve \eqref{FBSDE-power-2} under a weak interaction assumption. In order to avoid any unreasonable assumptions on $(h,\gamma,\sigma,\sigma^0)$, we need to drop the non-homogenous terms without $\theta$. This can be done by considering the difference. Second, \eqref{MF-FBSDE-diff} is the starting point of the asymptotic expansion result in Section \ref{sec:asymptotic}.
\end{remark}

\subsection{The Equivalent BSDE}\label{sec:equivalent-BSDE}
In this section, we will show that the FBSDE \eqref{MF-FBSDE-diff} is equivalent to some BSDE.  
First, we will prove a general result on the equivalence between a class of FBSDEs and BSDEs. Subsequently, we will show that our FBSDE \eqref{MF-FBSDE-diff} is a special case and thus can be transformed into some BSDE.

We consider the following FBSDE
\begin{equation}\label{general-FBSDE-transformation}
	\left\{\begin{split}
		d\mathcal X_t=&~\textit{dr}(t,\mathcal Y_t,\mathcal Z_t,\mathcal Z^0_t)\,dt+\textit{ diff}(t,\mathcal Y_t,\mathcal Z_t,\mathcal Z^0_t)dW_t+{\it diff}^0(t,\mathcal Z_t,\mathcal Z^0_t)\,dW^0_t,\\
		-d\mathcal Y_t=&~{\it Dr}(t,\mathcal Y_t,\mathcal Z_t,\mathcal Z^0_t)\,dt-\mathcal Z_t\,dW_t-\mathcal Z^0_t\,dW^0_t,\\
		\mathcal X_0=&~\bm{x},\quad \mathcal Y_T=\nu\mathbb E[\mathcal X_T|\mathcal F^0_T],
	\end{split}\right.
\end{equation}
where the coefficients are assumed to satisfy the following conditions.
\begin{itemize}
	\item {\bf AS-1.}  For any given  $\widetilde{\mathcal Y}$, $\mathcal Z$ and $\mathcal Z^0$, the following mean field SDE for $\mathcal Y$ has a unique solution
	\begin{equation}\label{condition:Y}
		\begin{split}
		\widetilde{\mathcal Y}_t=\mathcal Y_t-\int_0^t\nu\mathbb E\left[	{\it dr}(s,\mathcal Y_s,\mathcal Z_s,\mathcal Z^0_s)|\mathcal F^0_s	\right]\,ds-\int_0^t\nu\mathbb E\left[{\it diff}^0(s,\mathcal Z_s,\mathcal Z^0_s)|\mathcal F^0_s	\right]\,dW^0_s.
		\end{split}
	\end{equation}
The unique solution is denoted by $\mathcal Y_t=g_1(t, \widetilde{\mathcal Y}_{\cdot\wedge t},\mathcal Z_{\cdot\wedge t},\mathcal Z^0_{\cdot\wedge t} )$.
\item {\bf AS-2.} For each given $\widetilde{\mathcal Z}^0$ and $\mathcal Z$, the following equation for $\mathcal Z^0$ is uniquely solvable 
\begin{equation}\label{condition:Z}
	\widetilde{\mathcal Z}^0_t=\mathcal Z^0_t-	 \nu\mathbb E[ {\it diff}^0(t,\mathcal Z_t,\mathcal Z^0_t)|\mathcal F^0_t].
\end{equation}
The unique solution is denoted by $\mathcal Z^0_t=g_2(t,\mathcal Z_t,\widetilde{\mathcal Z}^0_t)$.
\end{itemize}

The next proposition shows that the solution to the FBSDE \eqref{general-FBSDE-transformation} satisfying the above conditions has a one-to-one correspondence with the solution to some BSDE. Such correspondence relies on the fact that $\mathcal Y$ depends on $\mathcal X$ in a linear way and only through the terminal condition, as well as the unique solvability of \eqref{condition:Y} and \eqref{condition:Z}. The linear dependence allows us to rewrite and split the terminal value into an  integral on $[0,t]$ and an integral on $[t,T]$, where the former integral can be merged into the solution to a new BSDE and the latter integral can be merged into the coefficient of the new BSDE.	The unique solvability of \eqref{condition:Y} and \eqref{condition:Z} yields the one-to-one correspondence.	
\begin{proposition}\label{prop:transformation}
Under {\bf AS-1} and {\bf AS-2}, there is a one-to-one correspondence between the solution to the FBSDE \eqref{general-FBSDE-transformation} and the solution to the following BSDE  
\begin{equation}\label{general-equivalent-BSDE}
	\begin{split}
		\widetilde{\mathcal Y}_t
		=&~\nu\mathbb E[\bm{x}]-\int_t^T\widetilde{\mathcal Z}_s\,dW_s-\int_t^T\widetilde{\mathcal Z}^0_s\,dW^0_s\\
		&~+ \int_t^T\bigg\{\nu\mathbb E\left[\left.{\it dr}\left(s,g_1(s,\widetilde{
			\mathcal Y}_{\cdot\wedge s}, \widetilde{\mathcal Z}_{\cdot\wedge s},g_2(\cdot\wedge s,\widetilde{\mathcal Z}_{\cdot\wedge s},\widetilde{\mathcal Z}^0_{\cdot\wedge s})),\widetilde{\mathcal Z}_s,g_2(s,\widetilde{\mathcal Z}_s,\widetilde{\mathcal Z}^0_s)\right)\right|\mathcal F^0_s  \right]\\
		&~+{\it Dr}\left(s,g_1(s,\widetilde{
			\mathcal Y}_{\cdot\wedge s}, \widetilde{\mathcal Z}_{\cdot\wedge s},g_2(\cdot\wedge s,\widetilde{\mathcal Z}_{\cdot\wedge s},\widetilde{\mathcal Z}^0_{\cdot\wedge s})),\widetilde{\mathcal Z}_s,g_2(s,\widetilde{\mathcal Z}_s,\widetilde{\mathcal Z}^0_s)\right)\bigg\}\,ds.
	\end{split}
\end{equation}
Let the solution to \eqref{general-FBSDE-transformation} and the solution to \eqref{general-equivalent-BSDE} be $(\mathcal X,\mathcal Y,\mathcal Z,\mathcal Z^0)$ and $(\widetilde{\mathcal Y},\widetilde{\mathcal Z},\widetilde{\mathcal Z}^0)$, respectively.
The relationship is given by for each $t\in[0,T]$
\begin{equation}\label{relation-widetilde-XYZ}
	\left\{\begin{split}
		\widetilde{\mathcal Y}_t=&~\mathcal Y_t- \int_0^t\nu\mathbb E[{\it dr}(s,\mathcal Y_s,\mathcal Z_s,\mathcal Z^0_s)|\mathcal F^0_s  ]\,ds-\int_0^t \nu\mathbb E[ {\it diff}^0(s,\mathcal Z_s,\mathcal Z^0_s)|\mathcal F^0_s ]	\,dW^0_s\\
		\widetilde{\mathcal Z}_t=&~\mathcal Z_t\\
		\widetilde{\mathcal Z}^0_t=&~\mathcal Z^0_t-	 \nu\mathbb E[ {\it diff}^0(t,\mathcal Z_t,\mathcal Z^0_t)|\mathcal F^0_t ].
	\end{split}\right.
\end{equation}
\end{proposition}
\begin{proof}
	From the dynamics of \eqref{general-FBSDE-transformation}, we have the following:
	\begin{align*}
			\mathcal Y_t=&~\nu\mathbb E[\mathcal X_T|\mathcal F^0_T]+\int_t^T {\it Dr}(s,\mathcal Y_s,\mathcal Z_s,\mathcal Z^0_s)\,ds-\int_t^T\mathcal Z_s\,dW_s-\int_t^T\mathcal Z^0_s\,dW^0_s\\
			=&~\nu\mathbb E[\bm{x}]+ \int_0^T\nu\mathbb E[{\it dr}(s,\mathcal Y_s,\mathcal Z_s,\mathcal Z^0_s)|\mathcal F^0_s  ]\,ds +\int_0^T \nu\mathbb E[ {\it diff}^0(s,\mathcal Y_s,\mathcal Z_s,\mathcal Z^0_s)|\mathcal F^0_s ]	\,dW^0_s\\
			&~+\int_t^T{\it Dr}(s,\mathcal Y_s,\mathcal Z_s,\mathcal Z^0_s)\,ds-\int_t^T\mathcal Z_s\,dW_s-\int_t^T\mathcal Z^0_s\,dW^0_s\\
			=&~\nu\mathbb E[\bm{x}]+ \int_0^t\nu\mathbb E[{\it dr}(s,\mathcal Y_s,\mathcal Z_s,\mathcal Z^0_s)|\mathcal F^0_s  ]\,ds +\int_0^t \nu\mathbb E[ {\it diff}^0(s,\mathcal Z_s,\mathcal Z^0_s)|\mathcal F^0_s ]	\,dW^0_s\\
			&~+ \int_t^T\left\{\nu\mathbb E[{\it dr}(s,\mathcal Y_s,\mathcal Z_s,\mathcal Z^0_s)|\mathcal F^0_s  ]+{\it Dr}(s,\mathcal Y_s,\mathcal Z_s,\mathcal Z^0_s)\right\}\,ds\\
			&~-\int_t^T\mathcal Z_s\,dW_s-\int_t^T\left\{ \mathcal Z^0_s-	 \nu\mathbb E[ {\it diff}^0(s,\mathcal Z_s,\mathcal Z^0_s)|\mathcal F^0_s ]		\right\}\,dW^0_s.
	\end{align*}
Define for each $t\in[0,T]$
\begin{equation}\label{widetilde-Y}
	\begin{split}
	\widetilde{\mathcal Y}_t=&~\mathcal Y_t- \int_0^t \nu\mathbb E[{\it dr}(s,\mathcal Y_s,\mathcal Z_s,\mathcal Z^0_s)|\mathcal F^0_s  ]\,ds -\int_0^t \nu\mathbb E[ {\it diff}^0(s,\mathcal Z_s,\mathcal Z^0_s)|\mathcal F^0_s ]\,dW^0_s
	\end{split}
\end{equation}
and
\begin{equation}\label{widetilde-Z}
\left\{\begin{split}
	\widetilde{\mathcal Z}_t=&~\mathcal Z_t\\
	\widetilde{\mathcal Z}^0_t=&~\mathcal Z^0_t-	 \nu\mathbb E[ {\it diff}^0(t,\mathcal Z_t,\mathcal Z^0_t)|\mathcal F^0_t ].
	\end{split}\right.
\end{equation}
From \eqref{widetilde-Y} and \eqref{widetilde-Z}, conditions \eqref{condition:Y} and \eqref{condition:Z} imply that
\begin{equation}\label{eq:general-Z0}
	\mathcal Z^0_t=g_2(t,\widetilde{\mathcal Z}_t,\widetilde{\mathcal Z}^0_t)
\end{equation}
and
\begin{equation}\label{eq:general-Y}
	\mathcal Y_t=g_1(t, \widetilde{\mathcal Y}_{\cdot\wedge t},\mathcal Z_{\cdot\wedge t},\mathcal Z^0_{\cdot\wedge t} )=g_1(t,\widetilde{
	\mathcal Y}_{\cdot\wedge t}, \widetilde{\mathcal Z}_{\cdot\wedge t},g_2(\cdot\wedge t,\widetilde{\mathcal Z}_{\cdot\wedge t},\widetilde{\mathcal Z}^0_{\cdot\wedge t})).
\end{equation}
Then, $(\widetilde{\mathcal Y},\widetilde{\mathcal Z},\widetilde{\mathcal Z}^0)$ satisfies the BSDE \eqref{general-equivalent-BSDE}. 

Moreover, if there exists some $(\widetilde{\mathcal Y},\widetilde{\mathcal Z},\widetilde{\mathcal Z}^0)$ satisfying \eqref{general-equivalent-BSDE}, by equations \eqref{widetilde-Z}, \eqref{eq:general-Z0} and \eqref{eq:general-Y}, we can construct $(\mathcal X,\mathcal Y,\mathcal Z,\mathcal Z^0)$ satisfying \eqref{general-FBSDE-transformation}.
\end{proof}
We introduce the following BSDE
\begin{equation}\label{MF-BSDE}
	\begin{split}
		\widetilde Y_t=&~\theta\mathbb E[\log x]+\int_t^T\left\{ \mathcal J_1(s,\widetilde Z_s,\widetilde Z^0_s)+\mathcal J_2(s;\widetilde Z,\widetilde Z^0,\theta)\right\} \,ds-\int_t^T\widetilde Z_s\,dW_s-\int_t^T\widetilde Z^0_s\,dW^0_s,
	\end{split}
\end{equation}
where $\mathcal J_1(\cdot;\widetilde Z,\widetilde Z^0)$ are terms that do not depend on $\theta$
\begin{equation*}
	\begin{split}
		\mathcal J_1(\cdot,\widetilde Z,\widetilde Z^0)=&~\frac{1}{2}\left\{ 1+\frac{\gamma\sigma^2}{(1-\gamma)( \sigma^2+(\sigma^0)^2 )}	\right\}\widetilde  Z^2+\frac{1}{2}\left\{  1+\frac{\gamma(\sigma^0)^2}{ (1-\gamma)( \sigma^2+(\sigma^0)^2  ) }   \right\}(\widetilde Z^0)^2\\
		&~+\left\{	Z^{\bar o}+\frac{\gamma\sigma(h+\sigma Z^{\bar o}+\sigma^0 Z^{0,\bar o})}{(1-\gamma)(\sigma^2+(\sigma^0)^2)}				\right\}\widetilde Z+\left\{  Z^{0,\bar o} + \frac{\gamma\sigma^0( h+\sigma Z^{\bar o}+\sigma^0 Z^{0,\bar o}  )}{(1-\gamma)(\sigma^2+(\sigma^0)^2)}  \right\}\widetilde Z^0\\
		&~+\frac{\gamma\sigma\sigma^0}{(1-\gamma)(\sigma^2+(\sigma^0)^2)}\widetilde Z\widetilde Z^0,
	\end{split}
\end{equation*}
and $\mathcal J_2(\cdot;\widetilde Z,\widetilde Z^0,\theta)$ are terms that depend on $\theta$. The expression of $\mathcal J_2$ is cumbersome and we summarize it in Appendix \ref{sec:J2} due to the convenience of the statement in the main text. 

As a corollary of Proposition \ref{prop:transformation}, we can show that the FBSDE \eqref{MF-FBSDE-diff} and the BSDE \eqref{MF-BSDE} are equivalent.

\begin{corollary}\label{lem:equiv-BSDE}
	The wellposedness of the FBSDE \eqref{MF-FBSDE-diff} is equivalent to the wellposedness of the BSDE \eqref{MF-BSDE}.  The relationship is given by for each $t\in[0,T]$
	\begin{equation}\label{eq:widetilde-Y}
		\begin{split}
			\widetilde Y_t=&~ \overline Y_t+\theta\gamma \int_0^t \mathbb E\left[\left.    \frac{	\sigma_s\overline Z_s+\sigma^0_s\overline Z^0_s	}{(1-\gamma)(\sigma^2_s+(\sigma^0_s)^2	)}	\left(	h_s-	\frac{1}{1-\gamma}(	h_s+\sigma_sZ^{\bar o}_s+\sigma^0_sZ^{0,\bar o}_s	)	\right)			\right|\mathcal F^0_s				\right]\,ds\\
			&~-\theta\gamma\int_0^t  \mathbb E\left[ \left. \frac{	\left(	\sigma_s\overline Z_s+\sigma^0_s\overline Z^0_s	\right)^2	}{2(1-\gamma)^2(	\sigma^2_s+(\sigma^0_s)^2	)} \right|\mathcal F^0_s\right]\,ds\\
			&~	+\theta\gamma \int_0^t  \mathbb E\left[\left.	 \frac{h_s+\sigma_sZ^{\bar o}_s+\sigma^0_sZ^{0,\bar o}_s}{ (1-\gamma)( \sigma^2_s+(\sigma^0_s)^2 )  } 	\left(	h_s-\frac{h_s+\sigma_sZ^{\bar o}_s+\sigma^{0}_sZ^{0,\bar o}_s}{2(1-\gamma)}	\right)	\right|\mathcal F^0_s\right]  \,ds			\\
			&~+\theta\gamma \int_0^t\left\{\mathbb E\left[\left.	\frac{ \sigma_s\overline Z_s+\sigma^0_s\overline Z^0_s }{(1-\gamma)( \sigma^2_s+(\sigma^0_s)^2 )}\sigma^0_s	\right|\mathcal F^0_s\right]+\mathbb E\left[\left. \frac{h_s+\sigma^0_sZ^{\bar o}_s+\sigma^0_sZ^{0,\bar o}_s		}{(1-\gamma)(  \sigma^2_s+(\sigma^0_s)^2 )}\sigma^0_s	\right|\mathcal F^0_s	\right] \right\} \,dW^0_s
		\end{split}
	\end{equation}
and
	\begin{equation}\label{tranform-Z}
	\widetilde Z_t=\overline Z_t,\quad \widetilde Z^0_t=\overline Z^0_t+\theta\gamma\mathbb E\left[\left. \frac{ \sigma\overline Z_t+\sigma^0_t\overline Z^0_t}{(1-\gamma)(  \sigma^2_t+(\sigma^0_t)^2  )}\sigma^0_t\right|\mathcal F^0_t  \right]+\theta\gamma\mathbb E\left[\left.\frac{  h_t+\sigma Z^{\bar o}_t+\sigma^0_tZ^{0,\bar o}_t }{(1-\gamma)( \sigma^2_t+(\sigma^0_t)^2  )}\sigma^0_t\right|\mathcal F^0_t\right].
\end{equation}
\end{corollary}
\begin{proof}
	Let $\nu=(-\theta\gamma,-\theta\gamma)$ and $\mathcal X=(\overline X,X^{\bar o})^\top$, where $\overline X$ and $X^{\bar o}$ satisfy the forward dynamics of \eqref{MF-FBSDE-diff} and \eqref{MF-FBSDE-0}, respectively. Because the drift of $\mathcal X$ does not depend on the $\mathcal Y$-component, \textbf{AS-1} trivially holds. To verify \textbf{AS-2}, it is sufficient to solve a unique $\overline Z^0$ from \eqref{tranform-Z}. 
	

	Multiplied by $\frac{(\sigma^0)^2}{(1-\gamma)(\sigma^2+(\sigma^0)^2 )}$ on both sides of the second equality in \eqref{tranform-Z} and taking conditional expectations $\mathbb E[\cdot|\mathcal F^0_t]$, we obtain an equality for $\mathbb E\left[\left.\frac{(\sigma^0_t)^2}{(1-\gamma)(\sigma_t^2+(\sigma^0_t)^2 )}\widetilde Z^0_t\right|\mathcal F^0_t  \right]$ in terms of $\mathbb E\left[\left. \frac{(\sigma^0_t)^2}{(1-\gamma)(\sigma^2_t+(\sigma^0_t)^2 )} \overline Z^0_t\right|\mathcal F^0_t \right]$, from which we get for each $t\in[0,T]$
	\begin{equation}\label{transform-Z-2}
		\begin{split}
			&~\mathbb E\left[\left.\frac{(\sigma^0_t)^2}{(1-\gamma)(\sigma_t^2+(\sigma^0_t)^2 )}\overline Z^0_t\right|\mathcal F^0_t  \right]\\
			=&~\frac{\mathbb E\left[\left.\frac{(\sigma^0_t)^2}{(1-\gamma)(\sigma_t^2+(\sigma^0_t)^2 )}\widetilde Z^0_t\right|\mathcal F^0_t \right]-\mathbb E\left[\left.\frac{\theta\gamma(\sigma^0_t)^2}{(1-\gamma)(\sigma_t^2+(\sigma^0_t)^2 )} \right|\mathcal F^0_t\right]\mathbb E\left[\left.\frac{\sigma_t\sigma^0_t}{ (1-\gamma)( \sigma^2_t+(\sigma^0_t)^2 ) }\widetilde Z_t \right|\mathcal F^0_t\right]  }{1+\mathbb E\left[	 \left.	 \frac{\theta\gamma(\sigma^0_t)^2}{ (1-\gamma)( \sigma^2_t+(\sigma^0_t)^2  ) }			\right|\mathcal F^0_t	\right]}\\
					&~-\frac{\mathbb E\left[\left.\frac{\theta\gamma(\sigma^0_t)^2}{(1-\gamma)(\sigma_t^2+(\sigma^0_t)^2 )} \right|\mathcal F^0_t \right]\mathbb E\left[\left. \frac{ h_t\sigma^0_t+\sigma_t\sigma^0_t Z^{\bar o}_t+(\sigma^0_t)^2 Z^{0,\bar o} }{ (1-\gamma)(\sigma^2_t+(\sigma^0_t)^2) } \right|\mathcal F^0_t\right]}{1+\mathbb E\left[\left. \frac{\theta\gamma(\sigma^0_t)^2}{ (1-\gamma)( \sigma^2_t+(\sigma^0_t)^2  ) }  \right|\mathcal F^0_t\right]}.
		\end{split}
	\end{equation}
	Taking \eqref{transform-Z-2} back into \eqref{tranform-Z} and rearranging terms, we obtain $\overline Z^0$ in terms of $\widetilde Z^0$ and $\widetilde Z$
\begin{equation}\label{eq:barZ0-intermsof-tildeZ0}
	\begin{split}
	\overline Z^0_t=&~\widetilde Z^0_t-\frac{\theta\gamma\mathbb E\left[\left.\frac{\sigma_t\sigma^0_t}{(1-\gamma)( \sigma^2_t+(\sigma^0_t)^2 )}\widetilde Z_t\right|\mathcal F^0_t\right]+\theta\gamma\mathbb E\left[\left.\frac{(\sigma^0_t)^2}{ (1-\gamma)( \sigma^2_t+(\sigma^0_t)^2 ) }\widetilde Z^0_t\right|\mathcal F^0_t\right]}{1+\mathbb E\left[\left. \frac{\theta\gamma(\sigma^0_t)^2}{ (1-\gamma)( \sigma^2_t+(\sigma^0_t)^2  ) }  \right|\mathcal F^0_t\right]}\\
	&~-\frac{\theta\gamma\mathbb E\left[\left.\frac{  h_t\sigma^0_t+\sigma_t\sigma^0_t Z^{\bar o}_t+(\sigma^0_t)^2 Z^{0,\bar o}_t    }{(1-\gamma)(  \sigma^2_t+(\sigma^0_t)^2  )} \right|\mathcal F^0_t\right]}{1+\mathbb E\left[\left. \frac{\theta\gamma(\sigma^0_t)^2}{ (1-\gamma)( \sigma^2_t+(\sigma^0_t)^2  ) }  \right|\mathcal F^0_t\right]},\quad t\in[0,T].
	\end{split}
	\end{equation}
Thus, \textbf{AS-2} is verified, and Proposition \ref{prop:transformation} implies that the FBSDE \eqref{MF-FBSDE-diff} is equivalent to the BSDE \eqref{MF-BSDE}.
\end{proof}

\subsection{Wellposedness of the BSDE \eqref{MF-BSDE} and  the FBSDE \eqref{FBSDE-power-2}}
The BSDE \eqref{MF-BSDE} is a quadratic one of conditional mean field type, and it does not satisfy the assumptions in \cite{Hibon-Hu-Tang-2017}. In particular, the quadratic growth in \eqref{MF-BSDE} comes from both $(\widetilde Z,\widetilde Z^0)$ and the conditional expectation of $(\widetilde Z,\widetilde Z^0)$; refer to the expression for $\mathcal J_2$ in Appendix \ref{sec:J2}. In addition, the feature of $\mathcal J_2$ is that it includes all mean field terms, and each term in $\mathcal J_2$ can be controlled by $\theta$. This observation motivates us to solve \eqref{MF-BSDE} under a weak interaction assumption, where the competition parameter $\theta$ is assumed to be sufficiently small.

To make it convenient to apply the theory of quadratic BSDE, by a change of measure, we can transform $\mathcal J_1$ into a pure quadratic term. Indeed, 
define
\begin{equation}\label{P-o}
	\begin{split}
	\frac{d\mathbb P^{\bar o}}{d\mathbb P}
	=&~\mathcal E\left(	\int_0^\cdot  \left\{ Z^{\bar o}_s+\frac{\gamma\sigma_s(h_s+\sigma_s Z^{\bar o}_s+\sigma^0_sZ^{0,\bar o}_s   )}{(1-\gamma)(\sigma^2_s+(\sigma^0_s)^2)} \right\}\,dW_s+\int_0^\cdot \left\{Z^{0,\bar o}_s+  \frac{\gamma\sigma^0_s( h_s+\sigma_s Z^{\bar o}_s+\sigma^0_s Z^{0,\bar o}_s )}{(1-\gamma)(\sigma^2_s+(\sigma^0_s)^2) } \right\}\,dW^0_s	\right)\\
	:=&~\mathcal E\left(\int_0^\cdot \mathcal M_s\,d(W_s,W^0_s)^\top \right),
	\end{split}
\end{equation}
where
\begin{equation}\label{def:M}
	\mathcal M=\left(  Z^{\bar o}+\frac{\gamma\sigma(h+\sigma Z^{\bar o}+\sigma^0Z^{0,\bar o}   )}{(1-\gamma)(\sigma^2+(\sigma^0)^2)},~Z^{0,\bar o}+  \frac{\gamma\sigma^0( h+\sigma Z^{\bar o}+\sigma^0 Z^{0,\bar o} )}{(1-\gamma)(\sigma^2+(\sigma^0)^2) }\right).
\end{equation}
By Proposition \ref{lem:YZ-o} and \cite[Theorem 2.3]{Kazamaki-2006}, it holds that $\mathbb P^{\bar o}$ is a probability measure, and the Girsanov theorem yields that
\begin{equation}\label{MF-BSDE-P-o}
	\widetilde Y_t=-\theta\gamma\mathbb E[\log(x)]+\int_t^T\left(\mathcal J^{\bar o}_1(s,\widetilde Z_s,\widetilde Z^0_s)+\mathcal J_2(s;\widetilde Z,\widetilde Z^0,\theta)\right)\,ds-\int_t^T\widetilde Z_s\,d\widetilde W_s-\int_t^T\widetilde Z^0_s\,d\widetilde W^0_s,
\end{equation}
where 
\begin{equation}\label{widetilde-W}
		\begin{split}
	&~W^{\bar o}_t:=(\widetilde W_t,\widetilde W^0_t)^\top
	=\\
	&~\left(W_t-\int_0^t\left\{Z^{\bar o}_s+\frac{\gamma\sigma_s(h_s+\sigma_s Z^{\bar o}_s+\sigma^0_sZ^{0,\bar o}_s   )}{(1-\gamma)(\sigma^2_s+(\sigma^0_s)^2)}\right\}\,ds,~W^0_t-\int_0^t\left\{Z^{0,\bar o}_s+  \frac{\gamma\sigma^0_s( h_s+\sigma_s Z^{\bar o}_s+\sigma^0_s Z^{0,\bar o}_s )}{(1-\gamma)(\sigma^2_s+(\sigma^0_s)^2) }\right\}\,ds\right)^\top
		\end{split}
\end{equation}
 is a two-dimensional Brownian motion under $\mathbb P^{\bar o}$, and
\begin{equation}\label{def:J1-o}
	\begin{split}
	\mathcal J^{\bar o}_1(\cdot,\widetilde Z,\widetilde Z^0)=&~\frac{1}{2}\left\{ 1+\frac{\gamma\sigma^2}{(1-\gamma)( \sigma^2+(\sigma^0)^2 )}	\right\}\widetilde  Z^2+\frac{1}{2}\left\{  1+\frac{\gamma(\sigma^0)^2}{ (1-\gamma)( \sigma^2+(\sigma^0)^2  ) }   \right\}(\widetilde Z^0)^2\\
	&~+\frac{\gamma\sigma\sigma^0}{(1-\gamma)(\sigma^2+(\sigma^0)^2)}\widetilde Z\widetilde Z^0.
	\end{split}
\end{equation}
Therefore, solving \eqref{MF-BSDE} under $\mathbb P$ is equivalent to solving \eqref{MF-BSDE-P-o} under $\mathbb P^{\bar o}$. To do so, we will use a fixed point argument as in \cite{Tevzadze2008} to study a general BSDE with \eqref{MF-BSDE-P-o} as a special case. 
\begin{lemma}\label{wellposedness-general-QBSDE}
 Define a BSDE 
	\begin{equation}\label{eq:general-QBSDE}
		-d\widehat{\mathcal Y}_t=\left\{		f_1(t,\widehat{\mathcal Z}_t)+f_2(t;\widehat{\mathcal Z},\theta)	\right\}\,dt-\widehat{\mathcal Z}_t\,d  W^{\bar o}_t,\quad\mathcal Y_T=\theta\bf{x},
	\end{equation}
The random coefficients $f_1$ and $f_2$ are assumed to satisfy the following conditions: there exists an $\mathcal A$-measurable, positive and bounded random variable $c_1$, an increasing positive locally bounded function $C_1$, and a positive constant $C_2$ such that for any $\widehat z, \widehat z'\in H^2_{BMO,\mathbb P^{\bar o}}$ it holds that:

\textbf{ASS-1.} \quad
$
|f_1(t,\widehat z)|\leq c_1|\widehat z|^2
$ \qquad and\qquad 
	$\left\|	 \left|	 2c_1 f_2(\cdot;\widehat{z},\theta)	\right|^{\frac{1}{2}}			\right\|_{BMO,\mathbb P^{\bar o}}^2\leq \|\theta\|C_1(\|\widehat z\|_{BMO,\mathbb P^{\bar o}});
$

\textbf{ASS-2.}  $
	|f_1(t,\widehat z)	-f_1(t,\widehat z' )	|\leq c_1|\widehat z-\widehat z'|( |\widehat z|+|\widehat z'| )$, and  $\|	f_2(\cdot;\widehat z,\theta)-f_2(\cdot;\widehat z',\theta)	\| _{BMO,\mathbb P^{\bar o}}\leq \|\theta\| C_2\|\widehat z-\widehat z'  \|_{BMO,\mathbb P^{\bar o}} \left(	1+	\|\widehat z \|_{BMO,\mathbb P^{\bar o}}	+\|\widehat z'  \|_{BMO,\mathbb P^{\bar o}}	\right).
$

Then, for each fixed $0<R\leq\frac{1}{4\sqrt{2}\|c_1\|}$, there exists a constant $\theta^*$ that is small enough and only depends on $R$, ${\bf x}$, $c_1$, $C_1$ and $C_2$, such that for each $0\leq \|\theta\|\leq\theta^*$
the BSDE \eqref{eq:general-QBSDE} admits a unique solution $(\widehat{\mathcal Y},\widehat{\mathcal Z})$ with $\widehat{\mathcal Z}$ located in the $R$-ball of $H^2_{BMO,\mathbb P^{\bar o}}$. Furthermore, the solution satisfies the following estimate
\begin{equation}\label{estimate:Y-general-QBSDE}
	\|\widehat{\mathcal Y}\|_\infty\leq \|\theta{\bf x}\| - \frac{1}{2\underline{c_1}}\log\bigg\{ 1-\left\|	\left|2 c_1 f_2(\cdot;\widehat{\mathcal Z},\theta) \right|^{\frac{1}{2}}	\right\|^2_{BMO,\mathbb P^{\bar o}}		\bigg\}
\end{equation}
and
\begin{equation}\label{estimate:Z-general-QBSDE}
	\|\widehat{\mathcal Z}\|^2_{BMO,\mathbb P^{\bar o}}\leq  \frac{\frac{1}{2\underline{ c_1}}e^{2\|c_1\| \|\theta{\bf x}\|}-\frac{1}{2\underline {c_1}}-\|\theta{\bf x}\|}{\underline {c_1}}+\frac{e^{2\|c_1\| \|\widehat{\mathcal Y}\|_\infty}-1}{\underline{c_1}}\left\|  |f_2(\cdot;\widehat{\mathcal Z},\theta)|^{\frac{1}{2}} \right\|^2_{BMO,\mathbb P^{\bar o}},
\end{equation}
where we recall that $\underline{c_1}$ is the essential infimum of $c_1$.
\end{lemma}
\begin{proof}
{\bf Step 1.} In the first step, we prove that for each fixed $R>0$ and $\widehat z\in H^2_{BMO,\mathbb P^{\bar o}}$ with $\|\widehat z\|^2_{BMO,\mathbb P^{\bar o}}\leq R$, there exists a unique solution to
\begin{equation}\label{eq:QBSDE-z-fixed}
-d\widehat{\mathcal Y}_t=\left\{		f_1(t,\widehat{\mathcal Z}_t)+f_2(t;\widehat{z},\theta)	\right\}\,dt-\widehat{\mathcal Z}_t\,d W^{\bar o}_t,\quad\mathcal Y_T=\theta\bf{x}.
\end{equation}
In addition, the solution satisfies the estimates \eqref{estimate:Y-general-QBSDE} and \eqref{estimate:Z-general-QBSDE} with $\widehat{\mathcal Z}$ on the right side replaced by $\widehat z.$

In order to prove the claim in {\bf Step 1}, we choose $\theta_1^*$, such that $\theta_1^* C_1(R)<1$. Then for all $0\leq\|\theta\|\leq \theta_1^*$, we have by {\bf ASS-1}
\[
	\left\|	 \left|	 2c_1 f_2(\cdot;\widehat{z},\theta)	\right|^{\frac{1}{2}}			\right\|_{BMO,\mathbb P^{\bar o}}^2\leq C_1(R)\|\theta\|<1,
\] 
which implies that by \cite[Theorem 2.2]{Kazamaki-2006}
\[
	\mathbb E^{\mathbb P^{\bar o}}\left[\left. e^{ 2c_1\int_t^T f_2(s;\widehat z,\theta) }\,ds\right|\mathcal F_t		\right]\leq \frac{1}{ 1-  	\left\|	 \left|	 2c_1 f_2(\cdot;\widehat{z},\theta)	\right|^{\frac{1}{2}}			\right\|_{BMO,\mathbb P^{\bar o}}^2  }<\infty.
\]
Thus, all conditions in  \cite[Proposition 3]{BH-2008} are satisfied. It yields a solution $(\widehat{\mathcal Y},\widehat{\mathcal Z})$ of \eqref{eq:QBSDE-z-fixed}, such that the estimate for  $\widehat{\mathcal Y}$ holds. In order to obtain the estimate for $\widehat{\mathcal Z}$, we define 
$\mathcal T(y)=\frac{\frac{1}{2c}e^{2c_1|y|}-\frac{1}{2c_1}-|y|}{c_1}$. It\^o's formula implies that
\begin{equation*}
	\begin{split}
		\mathcal T(\widehat{\mathcal Y}_t)=&~\mathcal T(\theta{\bf x})+\int_t^T\mathcal T'(\widehat{\mathcal Y}_s)\left(f_1(s,\widehat {\mathcal Z}_s)+f_2(s;\widehat z,\theta)\right)\,ds\\
		&~-\frac{1}{2}\int_t^T\mathcal T''(\widehat{\mathcal Y}_s)|\widehat Z_s|^2\,ds-\int_t^T\mathcal T'(\mathcal Y_s)\widehat{\mathcal Z}_s\,d  W^{\bar o}_s\\
		\leq&~\mathcal T(\theta{\bf x})+\int_t^T\left(c_1|\mathcal T'(\widehat{\mathcal Y}_s)|-\frac{1}{2}\mathcal T''(\widehat{\mathcal Y}_s)\right)|\widehat{\mathcal Z}_s|^2\,ds\qquad(\textrm{by }\textbf{ASS{\rm-}1})\\
		&~+\int_t^T\mathcal T'(\widehat{\mathcal Y}_s)f_2(s;\widehat z,\theta)\,ds-\int_t^T\mathcal T'(\widehat{\mathcal Y}_s)\widehat{\mathcal Z}_s\,d W^{\bar o}_s\\
		=&~\mathcal T(\theta{\bf x})-\int_t^T|\widehat{\mathcal Z}_s|^2\,ds+\int_t^T\mathcal T'(\widehat{\mathcal Y}_s)f_2(s;\widehat z,\theta)\,ds-\int_t^T\mathcal T'(\widehat{\mathcal Y}_s)\widehat{\mathcal Z}_s\,dW^{\bar o}_s.
	\end{split}
\end{equation*}
Note that $\mathcal T$ is increasing for $y>0$. Therefore, for any stopping time $
\tau$, it holds that
\begin{equation}
	\begin{split}
		&~\mathbb E^{\mathbb P^{\bar o}}\left[\left.\int_\tau^T|\widehat{\mathcal Z}_s|^2\,ds \right|\mathcal F_\tau \right] \\
		\leq&~ \mathcal T( \|\theta {\bf x}\| )+\int_\tau^T\mathcal T'(\widehat{\mathcal Y}_s)f_2(s;\widehat z,\theta)\,ds\\
		\leq&~\frac{\frac{1}{2\underline {c_1}}e^{2\|c_1\| \|\theta {\bf x}\|}-\frac{1}{2\underline {c_1}}-\|\theta{\bf x}\|}{\underline {c_1}}+\frac{e^{2\|c_1\| \|\widehat{\mathcal Y}\|_\infty}-1}{\underline{ c_1}}\left\| | f_2(\cdot;\widehat z,\theta)|^{\frac{1}{2}} \right\|^2_{BMO,\mathbb P^{\bar o}},
	\end{split}
\end{equation}
which implies the estimate for $\widehat{\mathcal Z}$.

Moreover, for any two solutions $(\widehat{\mathcal Y},\widehat{\mathcal Z})$ and $(\widehat{\mathcal Y}',\widehat{\mathcal Z}')$ in $S^\infty\times H^2_{BMO,\mathbb P^{\bar o}}$, $(\Delta{\widehat{\mathcal Y}}, \Delta {\mathcal{\widehat Z}}):=(\widehat{\mathcal Y}-\widehat{\mathcal Y}',\mathcal {\widehat Z}-\mathcal{\widehat Z}')$ follows for some stochastic process $ L$
\begin{equation*}
	\Delta\widehat{\mathcal Y}_t=\int_t^TL_s\Delta\widehat{\mathcal Z}_s\,ds-\int_t^T\Delta\widehat{\mathcal Z}_s\,d W^{\bar o}_s,
\end{equation*}
where $|L|\leq c_1 (|\widehat{\mathcal Z}|+|\widehat{\mathcal Z}'|)$. Define $\mathbb Q'$ by $\frac{d\mathbb Q'}{d\mathbb P^{\bar o}}=\mathcal E\left( \int_0 L_s\,d W^{\bar o}_s \right)$. \cite[Theorem 2.3]{Kazamaki-2006} implies that $\mathbb Q'$ is a probability measure. Consequently, we can rewrite the above equation as
$
\Delta{\widehat{\mathcal Y}}_t=-\int_t^T\Delta{\widehat{ \mathcal Z}}_s\,d  W_s^{\mathbb Q'},
$
where $ W^{\mathbb Q'}=W^{\bar o}-\int_0L_s\,ds$ is a Brownian motion under $\mathbb Q'$. Obviously it holds that $\Delta \widehat{\mathcal Y}=\Delta\widehat{\mathcal Z}=0$ and the uniqueness result follows.

\textbf{Step 2.} For each fixed $R$ and each $\widehat z$ in the $R$-ball of $H^2_{BMO,\mathbb P^{\bar o}}$, {\bf Step 1} yields a unique $(\widehat{\mathcal Y},\widehat{\mathcal Z})$, such that \eqref{eq:QBSDE-z-fixed} holds. According to the estimate for $(\widehat{\mathcal Y},\widehat{\mathcal Z})$ in {\bf Step 1}, we can choose a small $\theta_2^*\leq\theta_1^*$ depending on $R$, ${\bf x}$, $c_1$, $C_1$ and $C_2$, such that for any $0\leq\|\theta\|\leq\theta_2^*$, it holds that $\|\widehat{\mathcal Z}\|^2_{BMO,\mathbb P^{\bar o}}\leq R$. Thus, the mapping $\widehat z\mapsto\widehat{\mathcal Z}$ from the $R$-ball of $H^2_{BMO,\mathbb Q}$ to itself is well-defined. It remains to prove that this mapping is a contraction.

For any two $\widehat z$ and $\widehat z'$ in the $R$-ball of $H^2_{BMO,\mathbb P^{\bar o}}$, {\bf Step 1} yields a unique solution $(\widehat{\mathcal Y},\widehat{\mathcal Z})$ and $(\widehat{\mathcal Y}',\widehat{\mathcal Z}')$ of \eqref{eq:QBSDE-z-fixed}, corresponding to $\widehat z$ and $\widehat z'$, respectively.  Applying It\^o's formula to the BSDE of $(\Delta{\widehat{\mathcal Y}}, \Delta {\mathcal{\widehat Z}}):=(\widehat{\mathcal Y}-\widehat{\mathcal Y}',\mathcal {\widehat Z}-\mathcal{\widehat Z}')$, by {\bf ASS-2} and Young's inequality we have 
\begin{equation*}
	\begin{split}
	&~	\|\Delta\widehat{\mathcal Y}	\|_\infty^2+\|\widehat{\mathcal Z}  \|_{BMO,\mathbb P^{\bar o}}^2\\
	\leq&~2\| \Delta\widehat{\mathcal Y} \|_\infty \|c_1\| \| \Delta\widehat{\mathcal Z} \|_{BMO,\mathbb P^{\bar o}} \left( \| \widehat{\mathcal Z} \|_{BMO,\mathbb P^{\bar o}}+\|\widehat{\mathcal Z}'\|_{BMO,\mathbb P^{\bar o}}			\right)\\
	&~+2C_2\|	\Delta\widehat{\mathcal Y}	\|_\infty\|\theta\|  \|\Delta\widehat{ z}\|_{BMO,\mathbb P^{\bar o}} \left(	1+\|\widehat z\|_{BMO,\mathbb P^{\bar o}} + \|\widehat z'\|_{BMO,\mathbb P^{\bar o}}		\right)\\
	\leq&~4 R\|c_1\|\|\Delta\widehat{\mathcal Y}\|_\infty\|\Delta\widehat{\mathcal Z}\|_{BMO,\mathbb P^{\bar o}} + 2C_2\|\theta\|(1+2R)\|\Delta \widehat{\mathcal Y}	\|_\infty\|	\Delta\widehat z	\|_{BMO,\mathbb P^{\bar o}}\\
	\leq&~ \frac{1}{2}\|\Delta\widehat{\mathcal Y}\|_\infty^2 + 8R^2\|c_1\|^2\|\Delta\widehat{\mathcal Z}\|^2_{BMO,\mathbb P^{\bar o}}+\frac{1}{2}\|\Delta\widehat{\mathcal Y}\|_\infty^2+2C_2^2\|\theta||^2(1+2R)^2\| \Delta\widehat z \|_{BMO,\mathbb P^{\bar o}}^2.
	\end{split}
\end{equation*}
First, choose $R$, such that $8R^2\|c_1\|^2\leq\frac{1}{4}$. Second, choose $\theta^*\leq\theta_2^*$, such that $2C_2^2(\theta^*)^2(1+2R)\leq\frac{1}{4}$. Then, for all $0\leq\|\theta\|\leq\theta^*$, it holds that $\|\Delta \widehat{\mathcal Z} \|_{BMO,\mathbb P^{\bar o}}\leq \frac{1}{2}\|\Delta\widehat z\|^2_{BMO,\mathbb P^{\bar o}}$, which implies a contraction. 
\end{proof}

As a corollary of Lemma \ref{wellposedness-general-QBSDE}, we get the wellposedness result of \eqref{MF-BSDE-P-o}.


\begin{theorem}\label{thm:wellpos-MF-BSDE}
Let \textbf{Assumption 1} hold. Let $c_1=\frac{1}{1-\gamma}$ and choose $R$, such that $R\leq \frac{1}{4\sqrt{2}\|c_1\|}$.
Then, there exists a positive constant $\theta^*$ only depending on $R,~T,~\gamma,~\sigma,~\sigma^0$ and $h$, such that for all $0\leq\|\theta\|\leq\theta^*$, there exists a unique $(\widetilde Y,\widetilde Z,\widetilde Z^0)\in S^\infty\times H^2_{BMO,\mathbb P^{\bar o}}\times H^2_{BMO,\mathbb P^{\bar o}}$ satisfying \eqref{MF-BSDE-P-o}, and the estimates \eqref{estimate:Y-general-QBSDE} and \eqref{estimate:Z-general-QBSDE}. 
\end{theorem}
\begin{proof}
It is sufficient to verify that the driver of \eqref{MF-BSDE-P-o} satisfies {\bf ASS-1} and {\bf ASS-2} in Lemma \ref{wellposedness-general-QBSDE}. Let $\widehat{\mathcal Z}=(\widetilde Z,\widetilde Z^0)$.
First, by the definition of $\mathcal J_1^{\bar o}$ in \eqref{def:J1-o}, it is straightforward to verify that
\[
	|\mathcal J^{\bar o}_1(t,\widehat{\mathcal Z})|\leq c_1|\widehat{\mathcal Z}|^2,\quad\textrm{  and  }\quad      |	\mathcal J_1^{\bar o}(t,\widehat{\mathcal Z}) -\mathcal J_1^{\bar o}(t,\widehat{\mathcal Z}')		|\leq c_1 | \widehat{\mathcal Z} - \widehat{\mathcal Z}' | (	|  \widehat{\mathcal Z} |+| \widehat{\mathcal Z}’  |	).
\]
Furthermore, by the expression of $\mathcal J_2$ in Appendix \ref{sec:J2} and \cite[Lemma A.1]{FSZ-2021}, there exists a positive constant $C_2$ that only depends on $\gamma$, $\sigma$, $\sigma^0$, $h$, $\|Z^{\bar o}\|_{BMO,\mathbb P^{\bar o}}$ and $\| Z^{0,\bar o} \|_{BMO,\mathbb P^{\bar o}}$, such that
\[
	\|	\mathcal J_2(t;\widehat{\mathcal Z},\theta)- \mathcal J_2(t;\widehat{\mathcal Z}',\theta)		\|_{BMO,\mathbb P^{\bar o}} \leq \|\theta\| C_2\| \widehat{\mathcal Z}-\widehat{\mathcal Z}' \|_{BMO,\mathbb P^{\bar o}}\left(	1+	\| \widehat{\mathcal Z}\|_{BMO,\mathbb P^{\bar o}}+\| \widehat{\mathcal Z}' \|_{BMO,\mathbb P^{\bar o}}	\right).
\]
Again using Appendix \ref{sec:J2} and \cite[Lemma A.1]{FSZ-2021}, there exists a positive locally bounded function $C_1$ depending on $\gamma$, $\sigma$, $\sigma^0$ and $h$, such that
\[
	\left\||	2c_1\mathcal J_2(\cdot;\widehat{\mathcal Z})		|^{\frac{1}{2}}	\right\|^2_{BMO,\mathbb P^{\bar o}}\leq \|\theta\| C_1(\|\widehat{\mathcal Z}\|_{BMO,\mathbb P^{\bar o}}  ).
\]
Therefore, {\bf ASS-1} and {\bf ASS-2} are satisfied. From Lemma \ref{wellposedness-general-QBSDE}, we obtain the desired results. 	
\end{proof}
The following corollary implies that the unique solution to \eqref{MF-BSDE} can be controlled by $\theta$. In particular, the triple $\left(\|\widetilde Y\|_\infty, \|\widetilde Z\|_{BMO}, \|\widetilde Z^0\|_{BMO}\right)$ goes to $0$ as $\|\theta\|$ goes to $0$. This result is used to establish the convergence result in Corollary \ref{coro:wellpos-FBSDEs}, which will be used in Section \ref{sec:asymptotic}.
\begin{corollary}\label{coro:YZ-go2-0}
Let \textbf{Assumption 1} hold and	let $(\widetilde Y, \widetilde Z,\widetilde Z^0)$ be the unique solution to \eqref{MF-BSDE}. Then, it holds that
	\[
		\lim_{\|\theta\|\rightarrow 0}\left(\|\widetilde Y\|_\infty+\|\widetilde Z\|_{BMO}+ \|\widetilde Z^0\|_{BMO}\right)=0.
	\]
	and
	\[
		\lim_{\|\theta\|\rightarrow 0} \left( \|\widetilde Z\|_{M^p}+\|\widetilde Z^0\|_{M^p}\right)=0.
	\]
\end{corollary}
\begin{proof}
	Note that $\widetilde Z$ and $\widetilde Z^0$ belong to the $R$-ball of $H^2_{BMO,\mathbb P^{\bar o}}$, where $R$ is independent of $\theta$ by Theorem \ref{thm:wellpos-MF-BSDE}. Thus, by \eqref{estimate:Y-general-QBSDE} and letting $\|\theta\|\rightarrow 0$, we get $\|\widetilde Y\|_\infty\rightarrow 0$. Taking this convergence into \eqref{estimate:Z-general-QBSDE}, we get the convergence of $\widehat{\mathcal Z}:=(\widetilde Z,\widetilde Z^0)$ in $H^2_{BMO,\mathbb P^{\bar o}}$. The convergence of $\widehat{\mathcal Z}$ in $H^2_{BMO}$ can be obtained by \cite[Lemma A.1]{HMP-2019}. The convergence in $M^p$ is obtained by the energy inequality; see \cite[P.26]{Kazamaki-2006}.
\end{proof}

Another corollary of Theorem \ref{thm:wellpos-MF-BSDE} is that the FBSDE \eqref{MF-FBSDE-diff}, and thus the FBSDE \eqref{FBSDE-power-2} are wellposed.
\begin{corollary}\label{coro:wellpos-FBSDEs}
Let \textbf{Assumption 1} hold and $0\leq \|\theta\|\leq \theta^*$, where $\theta^*$ is the constant determined in Theorem \ref{thm:wellpos-MF-BSDE}. The FBSDE \eqref{MF-FBSDE-diff} has a unique solution $(\overline X,\overline Y,\overline Z,\overline Z^0)\in \bigcap_{p>1}S^p\times\bigcap_{p>1} S^p\times H^2_{BMO}\times H^2_{BMO}$ with the convergence 
	\[
	\lim_{\|\theta\|\rightarrow 0}\left(	\| \overline Z \|_{M^p}+\| \overline Z^{0} \|_{M^p} \right) =0
	\]
	and
	the FBSDE \eqref{FBSDE-power-2} is wellposed in $\bigcap_{p>1}S^p\times\bigcap_{p>1} S^p\times H^2_{BMO}\times H^2_{BMO}$.
\end{corollary}
\begin{proof}
Theorem \ref{thm:wellpos-MF-BSDE} and Corollary \ref{lem:equiv-BSDE} imply that there exists a unique $(\overline X,\overline Y,\overline Z,\overline Z^0)$ satisfying \eqref{MF-FBSDE-diff}.  
	Theorem \ref{thm:wellpos-MF-BSDE}, Proposition \ref{lem:YZ-o}, \eqref{tranform-Z}, \eqref{eq:barZ0-intermsof-tildeZ0}, and \cite[Lemma A.1]{HMP-2019} imply that $(\overline Z,\overline Z^0)\in H^2_{BMO}\times H^2_{BMO}$. 
In order to prove $\overline X\in\bigcap_{p>1} S^p$, we make a change of measure by defining
\[
	\frac{d\mathbb Q}{d\mathbb P}=\mathcal E\left(-\int_0^\cdot \widetilde\sigma^{\top}(\widetilde\sigma^{\top}\widetilde\sigma)^{-1}\left\{	h_s-\frac{h_s+\sigma_s Z^{\bar o}_s+\sigma^0_s Z^{0,\bar o}_s}{1-\gamma}	-\frac{\sigma_s\overline Z_s+\sigma^0_s\overline Z^0_s}{2(1-\gamma)}  \right\}\,d\overline W_s     \right):=\mathcal E\left(-\int_0^\cdot\widehat{\mathcal M}_s\,d\overline W_s\right),
\]
where we recall that $\widetilde\sigma=(\sigma,\sigma^0)^\top$.
Then, the dynamics of $\overline X$ can be rewritten as
\[
	d\overline X_t=\frac{	 \sigma_t\overline Z_t+ \sigma^0_t\overline Z^0_t	}{(1-\gamma)(\sigma^2_t+(\sigma^0_t)^2	)}\widetilde\sigma_t^\top\,d\widehat W_t,
\]
where
\[
	\widehat W_t=\overline W_t +\int_0^t\widetilde\sigma_s(\widetilde\sigma^\top_s\widetilde\sigma_s)^{-1}\left\{	h_s-\frac{h_s+\sigma_s Z^{\bar o}_s+\sigma^0_s Z^{0,\bar o}_s}{1-\gamma}	-\frac{\sigma_s\overline Z_s+\sigma^0_s\overline Z^0_s}{2(1-\gamma)}  \right\}\,ds
\]
is a $\mathbb Q$-Brownian motion. For each $p>1$, it holds that
\begin{equation*}
	\begin{split}
	\mathbb E^{\mathbb Q}\left[	\sup_{0\leq t\leq T}|\overline X_t|^p\right]\leq&~ C\mathbb E^{\mathbb Q}\left[\left(\int_0^T\left(	\frac{	 \sigma_t\overline Z_t+ \sigma^0_t\overline Z^0_t	}{(1-\gamma)(\sigma^2_t+(\sigma^0_t)^2	)}\right)^2\widetilde\sigma_t^\top\widetilde\sigma_t\,dt\right)^{\frac{p}{2}}		\right]	 \quad(\textrm{by BDG's inequality})\\
	\leq&~ C\left(\|\overline Z\|^2_{BMO,\mathbb Q}+\|\overline Z^0\|^2_{BMO,\mathbb Q}\right)\qquad(\textrm{by energy inequality})\\	<&~\infty \qquad (\textrm{by }\cite[\textrm{Lemma A.1}]{HMP-2019}).
	\end{split}
\end{equation*}
By the definition of $\mathbb Q$, for any $p>1$ and any $1<q<p_{\widehat{\mathcal M}}$, where we recall $p_{\widehat{\mathcal M}}$ is defined in Appendix B, it holds that
\begin{equation*}
	\begin{split}
		\mathbb E\left[	\sup_{0\leq t\leq T}| \overline X_t |^p	\right] =&~ \mathbb E^{\mathbb Q}\left[ \mathcal E\left(\int_0^T\widehat{\mathcal M}_t\,d\widehat W_t	\right)	\sup_{0\leq t\leq T}| \overline X_t |^p\right]\\
		\leq&~		\mathbb E^{\mathbb Q}\left[	\mathcal E\left(	\int_0^T\widehat{\mathcal M}_t\,d\widehat W_t		\right)^q\right]^{\frac{1}{q}} \mathbb E^{\mathbb Q}\left[	\sup_{0\leq t\leq T} |\overline X_t|^{pq^*}		\right]^{\frac{1}{q^*}}\qquad (\frac{1}{q}+\frac{1}{q^*}=1)\\
		<&~\infty\qquad (\textrm{by Lemma }\ref{lem:reverse}).
	\end{split}
\end{equation*}
Therefore, $\overline X\in\bigcap_{p>1}S^p$. The same arguement implies that $\overline Y\in\bigcap_{p>1}S^p.$
The convergence is obtained by Corollary \ref{coro:YZ-go2-0}, \eqref{tranform-Z}, and \eqref{eq:barZ0-intermsof-tildeZ0}. 

The wellposedness of \eqref{FBSDE-power-2} follows from \eqref{eq:diff_XYZ}, Proposition \ref{lem:YZ-o}, and Theorem \ref{thm:wellpos-MF-BSDE}.
\end{proof}


\subsection{Wellposedness of the MFG   \eqref{model-MFG-power}}

The main result in Section \ref{sec:wellposedness} is the following wellposedness result of MFG  \eqref{model-MFG-power}. In particular, our NE is unique.
\begin{theorem}\label{thm:wellposedness-MFG}
	Let \textbf{Assumption 1} hold and $0\leq\|\theta\|\leq \theta^*$, where $\theta^*$ is the constant determined in Theorem \ref{thm:wellpos-MF-BSDE}.
Let $(\widehat X^*,Y^*,Z^*,Z^{0,*})$ be the unique solution to \eqref{FBSDE-power-2}, and  $\mu^{*}_t=\exp\left(\mathbb E[ \widehat X^*_t|\mathcal F^0_t]\right)$, $t\in[0,T]$.  Under $\mu^{*}$, the unique optimal response for the representative player is
\begin{equation}\label{eq:optimal-action-in-thm-power}
	\pi^{*}=\frac{h+\sigma Z^*+\sigma^0Z^{0,*}}{(1-\gamma)(\sigma^2+(\sigma^0)^2)}
\end{equation}
and the value function given one realization of $(\theta,\gamma,x)$ follows
\[
	V^{}(\theta,\gamma,x_{},\mu^{*})=\frac{1}{\gamma}x_{}^\gamma e^{Y^*_0}.
\]
Moreover, $(\mu^{*},\pi^{*})$ is the unique NE of \eqref{model-MFG-power}, where $\log\mu^{*}\in \bigcap_{p>1} S^p_{\mathbb F^0}$ and $\pi^*\in  H^2_{BMO}$.

\end{theorem}
\begin{proof}
	By Proposition \ref{lem:NE-BSDE-power}, there is a one-to-one correspondence between the FBSDE \eqref{FBSDE-power-2} and the NE of \eqref{model-MFG-power}. Therefore, the existence and uniqueness result of NE can be obtained by Corollary \ref{coro:wellpos-FBSDEs}.
\end{proof}

\begin{remark}\label{rmk:dynamic-vf}
	By \cite[Proposition 15]{HIM-2005}, for each realization of the random variables $(\theta,\gamma,x_{})$, the value function satisfies the dynamic version 
	\begin{equation}\label{dynamic-vf}
	 V^{}(t,X^*_t;\theta,\gamma,x_{},\mu^{*})=\frac{1}{\gamma}(X^*_t)^\gamma e^{Y^{*}_t},
	\end{equation}
where
 $$V^{}(t,X^*_t;\theta,\gamma,x_{},\mu^{*}):=  \max_\pi\mathbb E\left[\left.	\frac{1}{\gamma}\Big(X_T(\mu^{*}_T)^{-\theta}\Big)^\gamma\right|\mathcal F_t	\right]. $$ The dynamic version of value function \eqref{dynamic-vf} will be used in Section \ref{sec:asymptotic}.
\end{remark}

\subsection{The Portfolio Game under $\mathcal A$-Measurable Market Parameters}\label{sec:constant-NE-MFG}
In this section, we consider a special case, in which the market parameters are $\mathcal A$-measurable. In addition to \textbf{Assumption 1}, we make the following assumption:

\textbf{Assumption 2.} 	For each $t\in[0,T]$, the return rate $h_t$ and the volatility $(\sigma_t,\sigma^0_t)$ are measurable w.r.t. $\mathcal A$.

 Under \textbf{Assumption 2}, we will construct a unique solution to the FBSDE \eqref{FBSDE-power-2} in closed form. As a result, we get an NE of the MFG \eqref{model-MFG-power} in closed form, which is proven to be unique in $L^\infty$. Furthermore, when all of the market parameters become time-independent, we revisit the model in \cite{LZ-2019}, and we prove that the constant equilibrium obtained in \cite{LZ-2019}  is the unique one in $L^\infty$, not only in the space of constant equilibria, as shown in \cite{LZ-2019}.

The following proposition shows the closed form solution to the FBSDE \eqref{FBSDE-power-2} under \textbf{Assumption 1} and {\bf Assumption 2}.
\begin{proposition}\label{prop:explicit-Zt}
	Under \textbf{Assumption 1} and \textbf{Assumption 2}, there exists a unique tuple $(\widehat X,Y,Z,Z^0)\in S^2\times S^2\times L^\infty\times L^\infty$ satisfying \eqref{FBSDE-power-2}. In particular, the $Z$-component of the solution has the following closed form expression
	\begin{equation}\label{explicit-Z}
		Z=0,\qquad Z^0=- \frac{\theta\gamma\mathbb E\left[ \frac{\sigma^0h}{(1-\gamma)( \sigma^2+(\sigma^0)^2   )} \right]}{1+\mathbb E\left[ \frac{\theta\gamma (\sigma^0)^2 }{  (1-\gamma)(  \sigma^2+(\sigma^0)^2   )    }  \right]}.
	\end{equation}
\end{proposition}
\begin{proof}
	We first verify that \eqref{explicit-Z}, together with some $\widehat X$ and $Y$, satisfies \eqref{FBSDE-power-2} by construction. Our goal is to construct $(Z,Z^0)$, such that $(Z_t,Z^0_t)$ is $\mathcal A$-measurable for each $t\in[0,T]$. Assuming that $(Z_t,Z^0_t)$ is $\mathcal A$-measurable for each $t\in[0,T]$ and taking the forward dynamics	in \eqref{FBSDE-power-2}  into the backward one in \eqref{FBSDE-power-2}, we have the following
	\begin{equation*}
		\begin{split}
			Y_t=&~-\theta\gamma\mathbb E[\log(x)]-\theta\gamma\int_0^T\mathbb E\left[ \frac{h_s+\sigma_sZ_s+\sigma^0_sZ^0_s}{(1-\gamma)(  \sigma^2_s+(\sigma^0_s)^2   )}	\left\{	h_s-\frac{h_s+\sigma_sZ_s+\sigma^0_sZ^0_s}{2(1-\gamma)}	\right\}	\right]\,ds\\
			&~+ \int_t^T  \left\{  \frac{Z^2_s+(Z^0_s)^2}{2}	+ \frac{	\gamma(  h_s+\sigma_sZ_s+\sigma^0_sZ^0_s   )^2		}{2(1-\gamma)(  \sigma^2_s+(\sigma^0_s)^2 )}		\right\}\,ds\\
			&~-\theta\gamma\int_0^T \mathbb E\left[ \frac{h_s+\sigma_sZ_s+\sigma^0_sZ^0_s	}{(1-\gamma)( \sigma^2_s+(\sigma^0_s)^2 )}\sigma^0_s	\right]\,dW^0_s-\int_t^TZ^0_s\,dW^0_s-\int_t^TZ_s\,dW_s,
		\end{split}
	\end{equation*}
where in the first line we use the assumption that $(Z_t,Z^0_t)$ is $\mathcal A$-measurable, and that $\mathcal A$ and $W^0$ are independent. In order to make $Y$ adapted, we let
\[
	Z\equiv 0,\qquad -\theta\gamma\mathbb E\left[ \frac{h+\sigma Z+\sigma^0Z^0	}{(1-\gamma)( \sigma^2+(\sigma^0)^2 )} \sigma^0	\right]-Z^0=0,
\]
which implies \eqref{explicit-Z}. Note that \eqref{explicit-Z} together with the constructed $Y$ is an adapted solution to \eqref{FBSDE-power-2}. 

Under \textbf{Assumption 2}, it holds that $Z^{\bar o}=Z^{0,\bar o}=0$. To prove the uniqueness result, by Corollary \ref{lem:equiv-BSDE}, it is sufficient to show that the solution to \eqref{MF-BSDE} is unique in $S^2\times L^\infty\times L^\infty$. Let $(\widetilde Y,\widetilde Z,\widetilde Z^{0})\in S^2\times L^\infty\times L^\infty$ and $(\widetilde Y',\widetilde Z',\widetilde Z^{0'})\in S^2\times L^\infty\times L^\infty$ be two solutions to \eqref{MF-BSDE}. Define
$
(\Delta Y, \Delta Z, \Delta Z^0	):=( \widetilde Y-\widetilde Y', \widetilde Z-\widetilde Z', \widetilde Z^0-\widetilde Z^{0'} ).
$
Then, by the dynamics of \eqref{MF-BSDE} and by noting that $(\widetilde Z,\widetilde Z^0)\in L^\infty\times L^\infty$ and $(\widetilde Z',\widetilde Z^{0'})\in L^\infty\times L^\infty$, the tuple $\Delta Y, \Delta Z, \Delta Z^0$ satisfies a conditional mean field BSDE with Lipschitz coefficients. 
Standard arguments imply that the unique solution is $(\Delta Y,\Delta Z,\Delta Z^0)=(0,0,0)$. 
\end{proof}


With the explicit solution in Proposition \ref{prop:explicit-Zt}, we can construct an optimal strategy in closed form for the representative player, which is unique in $L^\infty$. 
\begin{theorem}\label{thm:constant-equilibrium-MFG}
Let \textbf{Assumption 1} and \textbf{Assumption 2} hold.  Then, in $L^\infty$ the unique optimal response of the respresentative player is given by
\begin{equation}\label{eq:non-constant-strategy-power}
	\pi^{*}_t=\frac{h_t}{(1-\gamma)(\sigma^2_t+(\sigma^0_t)^2	)}	-\frac{\theta\gamma\sigma^0_t}{(1-\gamma)(\sigma^2_t+(\sigma^0_t)^2)}\frac{\mathbb E\left[	\frac{h_t\sigma^0_t}{(1-\gamma)(\sigma^2_t+(\sigma^0_t)^2)}	\right]}{1+\mathbb E\left[	\frac{\theta\gamma(\sigma^0_t)^2}{(1-\gamma)(	\sigma^2_t+(\sigma^0_t)^2	)}	\right]},\quad t\in[0,T].
\end{equation}
\end{theorem}
\begin{proof}
	Taking \eqref{explicit-Z} into \eqref{eq:optimal-action-in-thm-power}, we get \eqref{eq:non-constant-strategy-power}, which is unique in $L^\infty$ by Proposition \ref{prop:explicit-Zt}.
\end{proof}
As a corollary, when all coefficients become time-independent, we revisit the MFG in \cite{LZ-2019}.
\begin{corollary}[\textit{Lacker \& Zariphopoulou's MFGs revisited}]\label{corollary:LZ}
	Let \textbf{Assumption 1}  and  \textbf{Assumption 2} hold, and the return rate $h$ and the volatility $(\sigma,\sigma^0)$ be time-independent. Then, the unique optimal response is
	\begin{equation}\label{eq:constant-strategy-power}
		\pi^{*}=\frac{h}{(1-\gamma)(\sigma^2+(\sigma^0)^2	)}	-\frac{\theta\gamma\sigma^0}{(1-\gamma)(\sigma^2+(\sigma^0)^2)}\frac{\mathbb E\left[	\frac{h\sigma^0}{(1-\gamma)(\sigma^2+(\sigma^0)^2)}	\right]}{1+\mathbb E\left[	\frac{\theta\gamma(\sigma^0)^2}{(1-\gamma)(	\sigma^2+(\sigma^0)^2	)}	\right]}.
	\end{equation}
The constant equilibrium   \eqref{eq:constant-strategy-power}	is identical to \cite[Theorem 19]{LZ-2019}, which is unique in $L^\infty$. 
\end{corollary}

\begin{remark}
If $0\leq\|\theta\|\leq \theta^*$ holds, then \eqref{explicit-Z} is also unique in $H^2_{BMO}\times H^2_{BMO}$ by Corollary \ref{coro:wellpos-FBSDEs}. Therefore, the constant equilibrium \eqref{eq:constant-strategy-power} is the unique one in $H^2_{BMO}$, and there is no nonconstant equilibrium in the subspace of $H^2_{BMO}$ with \eqref{necessary:reverse-power} being true. 
\end{remark}

\section{Asymptotic Expansion in Terms of  $\theta$}\label{sec:asymptotic}
Motivated by the weak interaction assumption, we develop an asymptotic expansion result of the value function and the optimal investment. Specifically, we provide an approximation in any order of the value function and optimal investment for the model with competition in terms of the solution to the benchmark model without competition when the investor is only concerned with her own wealth.  Our idea is to start with the FBSDE characterization of the NE, and translate the expansion of the value function and the optimal investment to that of the solution to the FBSDE. 

Let $V^{\theta}$ and $\pi^\theta$ be the value function and the optimal investment of the MFG \eqref{model-MFG-power}, and $V^{\bar o}$ and $\pi^{\bar o}$ be the value function and the optimal investment of the benchmark utility maximization problem (i.e. when $\theta=0$ in \eqref{model-MFG-power}). 
Then by Theorem \ref{thm:wellposedness-MFG} and Remark \ref{rmk:dynamic-vf}, we have 
\begin{equation*}
	\begin{split}
	&~V^{\theta}_t=\frac{1}{\gamma}e^{ \gamma \widehat X^{\theta}_t+Y^{\theta}_t },\qquad \pi^\theta=\frac{h+\sigma Z^\theta+\sigma^0 Z^{0,\theta}}{(1-\gamma)(\sigma^2+(\sigma^0)^2)};\\
	&~ V^{\bar o}_t=\frac{1}{\gamma}e^{ \gamma X^{\bar o}_t+Y^{\bar o}_t  },\qquad \pi^{\bar o}=\frac{ h+\sigma Z^{\bar o}+\sigma^0 Z^{0,\bar o} }{(1-\gamma)(\sigma^2+(\sigma^0)^2)},
	\end{split}
\end{equation*}
where $(\widehat X^{\theta},Y^{\theta},Z^{\theta},Z^{0,\theta})$ and $(X^{\bar o},Y^{\bar o},Z^{\bar o},Z^{0,\bar o})$ are the unique solutions to \eqref{FBSDE-power-2} with $\theta$ and $\theta=0$, respectively. As a result, it holds that
\begin{equation*}
	\begin{split}
	\log\frac{V^{\theta}}{V_t^{\bar o}}=&~  \gamma(\widehat X^{\theta}_t-X^{\bar o}_t)  +(Y^{\theta}_t-Y^{\bar o}_t),\qquad \pi^\theta-\pi^{\bar o}=\frac{\sigma(Z^\theta-Z^{\bar o})+\sigma^0(Z^{0,\theta}-Z^{0,\bar o}) }{(1-\gamma)(\sigma^2+(\sigma^0)^2)}.
	\end{split}
\end{equation*}
Let $(\overline X^{\theta},  \overline Y^{\theta}, \overline Z^{\theta}, \overline Z^{0,\theta}):=( \widehat X^{\theta}-X^{\bar o},  Y^{\theta}- Y^{\bar o},  Z^{\theta}-Z^{\bar o},  Z^{0,\theta}-Z^{0,\bar o})$ be the unique solution to \eqref{MF-FBSDE-diff}. 

Our goal is to prove that  for any $n\geq 1$ there exist $(\widehat X^{(i)},  Y^{(i)},  Z^{(i)},  Z^{0,(i)})_{i=1,\cdots,n}$, such that it holds that\footnote{In Section \ref{sec:asymptotic}, $\theta^i$ means $\theta$ raised to the power of $i$. It should not be confused with player $i$'s competition parameter $\theta^i$ in the $N$-player game, e.g., \eqref{PU-intro}. }
\begin{equation*}\label{expansion-FBSDE}
	\begin{split}
	\widehat X^\theta-X^{\bar o}=&~\overline X^\theta=\sum_{i=1}^n\theta^i\widehat X^{(i)}+o(\theta^n),\qquad  Y^\theta-Y^{\bar o}=\overline Y^\theta=\sum_{i=1}^n\theta^i  Y^{(i)}+o(\theta^n),\\
	Z^\theta-Z^{\bar o}=&~\overline Z^\theta=\sum_{i=1}^n\theta^i  Z^{(i)}+o(\theta^n),\qquad Z^{0,\theta}-Z^{0,\bar o}=\overline Z^{0,\theta}=\sum_{i=1}^n\theta^i  Z^{0,(i)}+o(\theta^n).
	\end{split}
\end{equation*}

For $\varphi=\widehat X, Y, Z, Z^0$, define
\[
	  \varphi^{\theta,(1)}=\frac{\varphi^\theta-\varphi^{\bar o}}{\theta}=\frac{\overline\varphi^\theta}{\theta},
\]
which implies
\begin{equation}\label{widehat-XYZ-theta-1}
	\left\{\begin{split}
		d\widehat X^{\theta,(1)}_t=&~\bigg\{\frac{	\sigma_t  Z^{\theta,(1)}_t+\sigma^0_t  Z^{0,\theta,(1)}_t	}{(1-\gamma)(\sigma^2_t+(\sigma^0_t)^2	)}	\left(	h_t-	\frac{1}{1-\gamma}(	h_t+\sigma_tZ^{\bar o}_t+\sigma^0_tZ^{0,\bar o}_t	)	\right)	\\
		&~-\frac{\left(	\sigma_t\overline Z_t^\theta+\sigma^0_t\overline Z^{0,\theta}_t	\right)\left(	\sigma_t  Z^{\theta,(1)}_t+\sigma^0_t  Z^{0,\theta,(1)}_t	\right)	}{2(1-\gamma)^2(	\sigma^2_t+(\sigma^0_t)^2	)} \bigg\}\,dt \\
		&~+\frac{\sigma_t  Z^{\theta,(1)}_t+\sigma^0_t  Z^{0,\theta,(1)}_t}{(1-\gamma)(\sigma^2_t+(\sigma^0_t)^2)}  (\sigma_t\,dW_t+\sigma^0_t\,dW^0_t)  \\
		-d  Y^{\theta,(1)}_t=&~\left\{ \frac{( 2Z^{\bar o}_t+\overline Z^\theta_t )   Z^{\theta,(1)}_t}{2} +\frac{(2Z^{0,\bar o}_t+ \overline Z^{0,\theta}_t)  Z^{0,\theta,(1)}_t}{2} \right.\\
		&~\left. +  \frac{	\gamma(	2h_t+2\sigma_tZ^{\bar o}_t+2\sigma^0_tZ^{0,\bar o}_t+\sigma_t\overline Z^\theta_t+\sigma^0_t\overline Z^{0,\theta}_t		)(	\sigma_t  Z^{\theta,(1)}_t+\sigma^0_t  Z^{0,\theta,(1)}_t		)		}{2(1-\gamma)(\sigma^2_t+(\sigma^0_t)^2)}  \right\}\,dt\\
		&~-  Z^{\theta,(1)}_t\,dW_t-  Z^{0,\theta,(1)}_t\,dW^0_t,\\
		\widehat X^{\theta,(1)}_0=&~0,~  Y^{\theta,(1)}_T=-\gamma\mathbb E[\overline X^\theta_T|\mathcal F^0_T]-\gamma\mathbb E[X^{\bar o}_T|\mathcal F^0_T].
	\end{split}\right.
\end{equation}

We now establish the convergence of $(\widehat X^{\theta,(1)},  Y^{\theta,(1)},  Z^{\theta,(1)},  Z^{0,\theta,(1)})$ in a suitable sense as $\theta\rightarrow 0$.  Let the candidate limit of \eqref{widehat-XYZ-theta-1} satisfy
\begin{equation}\label{widehat-XYZ-1}
	\left\{\begin{split}
		d\widehat X^{(1)}_t=&~\frac{	\sigma_t  Z^{(1)}_t+\sigma^0_t  Z^{0,(1)}_t	}{(1-\gamma)(\sigma^2_t+(\sigma^0_t)^2	)}	\left(	h_t-	\frac{1}{1-\gamma}(	h_t+\sigma_tZ^{\bar o}_t+\sigma^0_tZ^{0,\bar o}_t	)	\right)\,dt	\\
		&~+\frac{\sigma_t  Z^{(1)}_t+\sigma^0_t  Z^{0,(1)}_t}{(1-\gamma)(\sigma^2_t+(\sigma^0_t)^2)}  (\sigma_t\,dW_t+\sigma^0_t\,dW^0_t)  \\
		-d  Y^{(1)}_t=&~\left\{  Z^{\bar o}_t   Z^{(1)}_t +Z^{0,\bar o}_t  Z^{0,(1)}_t \right.\\
		&~\left. +  \frac{	\gamma(	h_t+\sigma_tZ^{\bar o}_t+\sigma^0_tZ^{0,\bar o}_t	)(	\sigma_t  Z^{(1)}_t+\sigma^0_t  Z^{0,(1)}_t		)		}{(1-\gamma)(\sigma^2_t+(\sigma^0_t)^2})  \right\}\,dt\\
		&~-  Z^{(1)}_t\,dW_t-  Z^{0,(1)}_t\,dW^0_t,\\
		\widehat X^{(1)}_0=&~0,~  Y^{(1)}_T=-\gamma\mathbb E[X^{\bar o}_T|\mathcal F^0_T].
	\end{split}\right.
\end{equation}
%
The following lemma establishes the wellposedness result of \eqref{widehat-XYZ-1} and the convergence from \eqref{widehat-XYZ-theta-1} to \eqref{widehat-XYZ-1}. In particular, the first-order approximation of the FBSDE is obtained.
\begin{lemma}\label{lem:convergence-XYZtheta-(1)}
	Let \textbf{Assumption 1} hold.
	
	(1) There exists a unique $(\widehat X^{(1)},  Y^{(1)},  Z^{(1)},  Z^{0,(1)})\in \bigcap\limits_{p>1}S^p\times\bigcap\limits_{p>1} S^p\times\bigcap\limits_{p>1}M^p\times\bigcap\limits_{p>1}M^p$ satisfying \eqref{widehat-XYZ-1};
	
	(2) For each $p>1$, the following convergence holds
	\begin{equation}\label{eq:convergence-delta-YZ-theta}
	\lim\limits_{\theta\rightarrow 0} \left(\| \widehat X^{\theta,(1)}-\widehat X^{(1)} \|_{S^p}+\|   Y^{\theta,(1)}-  Y^{(1)} \|_{S^p}		+  \|  Z^{\theta,(1)}-  Z^{(1)}	\|_{M^p}+\|  Z^{0,\theta}-  Z^{0,(1)}\|_{M^p} \right)=0.
	\end{equation}
\end{lemma}
\begin{proof}
(1).	Recall $\mathbb P^{\bar o}$ defined in \eqref{P-o}. Then the backward dynamics in \eqref{widehat-XYZ-1} can be rewritten as
	\[
		d  Y^{(1)}_t=  Z^{(1)}_t\,d\widetilde W_t+  Z^{0,(1)}_t\,d\widetilde W^0_t,\qquad   Y^{(1)}_T=-\gamma\mathbb E[ X^{\bar o}_T|\mathcal F^0_T ],
	\]
	where $W^{\bar o}=(\widetilde W,\widetilde W^0)^\top$ is defined in \eqref{widetilde-W}. In order to apply \cite[Theorem 3.5]{Briand-C-2008}, it suffices to prove that for each $p>1$
	\[
		\mathbb E^{\mathbb P^{\bar o}}[ ( \mathbb E[ X^{\bar o}_T|\mathcal F^0_T ]   )^p ]<\infty.
	\]
	In fact, by the notation $\mathcal M$ in \eqref{def:M} and $p_{\mathcal M}$ in Appendix \ref{app:reverse}, for any $p>1$ and any $q\in(1, p_{\mathcal M})$
	\begin{equation}\label{app-reverse-1}
		\begin{split}
		&~\mathbb E^{\mathbb P^{\bar o}}[ ( \mathbb E[ |X^{\bar o}_T| |\mathcal F^0_T ]   )^p ]=\mathbb E\left[ \mathcal E_T(\mathcal M)( \mathbb E[ |X^{\bar o}_T| |\mathcal F^0_T ]   )^p \right	]\\
	\leq&~  \left(\mathbb E[ |\mathcal E_T(\mathcal M)|^{q}  ]\right)^{\frac{1}{q}} \left(\mathbb E\left[  \left( \mathbb E[ |X^{\bar o}_T|  |\mathcal F^0_T ]  \right)^{pq^*}  \right]\right)^{\frac{1}{q^*}}\qquad (\textrm{by H\"older's inequality,  }~\frac{1}{q}+\frac{1}{q^*}=1)\\
	\leq &~K^{\frac{1}{q}}(q,\|\mathcal M\|_{BMO})	\left(\mathbb E\left[  \left( \mathbb E[ |X^{\bar o}_T|  |\mathcal F^0_T ]  \right)^{pq^*}  \right]\right)^{\frac{1}{q^*}}				\qquad (\textrm{	 by Lemma \ref{lem:reverse}		})\\
	<&~\infty\qquad\qquad (	\textrm{by Proposition }\ref{lem:YZ-o}	). 
	\end{split}
	\end{equation}
Therefore, by \cite[Theorem 3.5]{Briand-C-2008}, there exists a unique $(   Y^{(1)},   Z^{(1)},  Z^{0,(1)})\in S^2_{\mathbb P^{\bar o}}\times\bigcap\limits_{p>1}M^p_{\mathbb P^{\bar o}}\times\bigcap\limits_{p>1}M^p_{\mathbb P^{\bar o}}$, which implies that $(   Y^{(1)},   Z^{(1)},  Z^{0,(1)})\in S^2_{}\times\bigcap\limits_{p>1}M^{p}\times\bigcap\limits_{p>1}M^{p}$. Indeed, by the definition of $\mathbb P^{\bar o}$ in \eqref{P-o}, it holds that
\[
	\frac{d\mathbb P}{d\mathbb P^{\bar o}}=\mathcal E\left(  -\int_0^\cdot \mathcal M_s\,dW^{\bar o}_s \right).
\]
Thus, by the same argument as \eqref{app-reverse-1}, we have for any $p>1$
\begin{equation*}
	\begin{split}
	&~\mathbb E\left[	\left(	\int_0^T (  Z^{(1)}_s)^2\,ds	\right)^{\frac{p}{2}}		\right]	  =\mathbb E^{\mathbb P^{\bar o}}\left[\mathcal E_T\left(	 -\int_0^\cdot \mathcal M_s\,dW^{\bar o}_s	\right)	\left(	\int_0^T (  Z^{(1)}_s)^2\,ds	\right)^{\frac{p}{2}}			\right]<\infty.
	\end{split}
\end{equation*}
The same result holds for $  Z^{0,(1)}$. By the dynamics of $  Y^{(1)}$ and standard estimate, we have $  Y^{(1)}\in\bigcap\limits_{p>1} S^p$.

(2). Let $\Delta \varphi^\theta=  \varphi^{\theta,(1)}- \varphi^{(1)}$ for $\varphi=\widehat X,Y,Z,Z^{0}$. Then, $(\Delta Y^\theta,\Delta Z^\theta, \Delta Z^{0,\theta})$ satisfies
\begin{equation}
	\left\{\begin{split}
		-d\Delta Y^\theta_t=&~\Bigg\{	\frac{1}{2}\overline Z^\theta_t  Z^{\theta,(1)}_t+\frac{1}{2}\overline Z^{0,\theta}_t  Z^{0,\theta,(1)}_t+		\frac{  \gamma\left( \sigma_t  Z^{\theta,(1)}_t+\sigma^0_t  Z^{0,\theta,(1)}_t \right )\left(	\sigma_t\overline Z^{\theta}_t+\sigma^0_t\overline Z^{0,\theta}_t	\right) }{2(1-\gamma)(\sigma^2_t+(\sigma^0_t)^2)}\\
		&~+\left(	Z^{\bar o}_t+\frac{	\gamma\sigma_t( h_t+\sigma_t Z^{\bar o}_t+\sigma^0_t Z^{0,\bar o}_t )		}{(1-\gamma)( \sigma^2_t+(\sigma^0_t)^2 )			}					\right)\Delta Z^\theta_t \\
		&~+\left( Z^{0,\bar o}+\frac{\gamma\sigma^0_t(h_t+\sigma_t Z^{\bar o}_t+\sigma^0_t Z^{0,\bar o}_t)}{(1-\gamma)(\sigma^2_t+(\sigma^0_t)^2)}			\right)\Delta Z^{0,\theta}_t\Bigg\}	\,dt
		-\Delta Z^\theta_t\,dW_t-\Delta Z^{0,\theta}\,dW^0_t,\\
		\Delta Y^\theta_T=&~-\gamma\mathbb E[ \overline X^\theta_T|\mathcal F^0_T  ]. 
	\end{split}\right.
\end{equation}
The conditions in \cite[Corollary 3.4]{Briand-C-2008} are satisfied. Indeed, by Corollary \ref{coro:wellpos-FBSDEs} and the energy inequality (see \cite[P. 26]{Kazamaki-2006}), it holds that for each $p>1$, $\mathbb E\left[\left(	\mathbb E[ \overline X^\theta_T|\mathcal F^0_T  ]		\right)^p	\right]<\infty$. The same reason, together with the result in (1), implies that for each $p>1$ the non-homogenous term in the driver of $\Delta Y^\theta$ is in $M^p$.
 Thus, \cite[Assumption A.3]{Briand-C-2008} is satisfied. \cite[Assumption A.2]{Briand-C-2008} holds due to Corollary \ref{coro:wellpos-FBSDEs} and Proposition \ref{lem:YZ-o}. 
	By \cite[Corollary 3.4]{Briand-C-2008}, it holds that for each $p>1$ there exists $p'>1$ and $p''>1$, such that
	\begin{equation*}
		\begin{split}
		&~	\| \Delta Y^\theta  \|_{S^p}+ \| \Delta Z^\theta \|_{M^p}	+ \|\Delta Z^{0,\theta}\|_{M^p}		\\
		\leq&~C\left\{ \mathbb E\left[ | \overline X^\theta_T |^{p'} \right]	+	( \|  Z^{\theta,(1)}	\|_{M^{p'}}+\|	  Z^{0,\theta,(1)}\|_{M^{p'}}	)(  \|	\overline Z^\theta	\|_{M^{p'}}+ \|	\overline Z^{0,\theta}	\|_{M^{p'}}	)	\right\}\\
		&~\times \Bigg\{	1+ \left\|  Z^{0,\bar o}+ \frac{ \gamma\sigma^0( h+\sigma Z^{\bar o}+\sigma^0 Z^{0,\bar o} ) }{(1-\gamma)(\sigma^2+(\sigma^0)^2)}	\right\|_{M^{p''}}+ \left\| Z^{\bar o} + \frac{\gamma\sigma(h+\sigma Z^{\bar o}+\sigma^0 Z^{0,\bar o})}{(1-\gamma)(\sigma^2+(\sigma^0)^2)}	\right\|_{M^{p''}}		\Bigg\}\\
		\rightarrow&~ 0,\qquad(\textrm{by Corollary }\ref{coro:wellpos-FBSDEs})
	\end{split}	
\end{equation*}
where $C$ does not depend on $\theta$.
\end{proof}

In order to establish higher order approximation, we define 
\begin{equation*}
\left\{	\begin{split}
	K^{(\theta,n-1),(\theta,1)}:=&~ \frac{\left(	\sigma  Z^{\theta,(n-1)}+\sigma^0  Z^{0,\theta,(n-1)} \right)\left(	\sigma  Z^{\theta,(1)}+\sigma^0  Z^{0,\theta,(1)}	\right)	}{2(	\sigma^2+(\sigma^0)^2	)},\\ K^{(i),(\theta,n-i)}:=&~\frac{  \left(	\sigma  Z^{(i)}+\sigma^0  Z^{0,(i)}	\right) \left(\sigma  Z^{\theta,(n-i)}+\sigma^0  Z^{0,\theta,(n-i)}	\right)  }{2( \sigma^2+(\sigma^0)^2 )},\\ K^{(i),(n-i)}:=&~\frac{  \left(	\sigma  Z^{(i)}+\sigma^0  Z^{0,(i)}	\right) \left(\sigma  Z^{(n-i)}+\sigma^0  Z^{0,(n-i)}	\right)  }{2( \sigma^2+(\sigma^0)^2 )}.
	\end{split}\right.
\end{equation*}
Based on the above notation, we now introduce the candidate $(\widehat X^{(n)},  Y^{(n)},  Z^{(n)}, Z^{0,(n)})$, $n\ge 2$. Intuitively, $(\widehat X^{(n+1)},  Y^{(n+1)},  Z^{(n+1)}, Z^{0,(n+1)})$ is the limit of  $$\left( \frac{\widehat X^{\theta,(n)}-\widehat X^{(n)}}{\theta}, \frac{  Y^{\theta,(n)}-  Y^{(n)}}{\theta}, \frac{  Z^{\theta,(n)}-  Z^{(n)}}{\theta}, \frac{  Z^{0,\theta,(n)}-  Z^{0,(n)}}{\theta}		\right):=(\widehat X^{\theta,(n+1)},  Y^{\theta,(n+1)},  Z^{\theta,(n+1)}, Z^{0,\theta,(n+1)}).$$
Thus, for each $n\geq 2$, we introduce two (decoupled) FBSDE systems iteratively
\begin{equation}\label{FBSDE-n}
	\left\{\begin{split}
		d\widehat X^{\theta,(n)}_t=&~\left\{	\frac{	\sigma_t  Z^{\theta,(n)}_t+\sigma^0_t  Z^{0,\theta,(n)}_t	}{(1-\gamma)( \sigma^2_t+(\sigma^0_t)^2 )}	\left\{	h_t-\frac{1}{1-\gamma}(	h_t+\sigma_tZ^{\bar o}_t+\sigma^0_tZ^{0,\bar o}_t	)	\right\}-\frac{K^{(\theta,n-1),(\theta,1)}_t}{(1-\gamma)^2}\right.	\\
		&~\left.-\sum_{j=1}^{n-2}  \frac{1}{(1-\gamma)^2} K^{(j),(\theta,n-j)}_t \right\}\,dt
		+\frac{\sigma_t  Z^{\theta,(n)}_t +\sigma^0_t  Z^{0,\theta,(n)}_t }{(1-\gamma)( \sigma^2_t+(\sigma^0_t)^2  )} (\sigma_t\,dW_t+\sigma^0_t\,dW^0_t)  \\
		-d  Y^{\theta,(n)}_t=&~\left\{\frac{ \gamma\left( \sigma_t  Z^{\theta,(n)}_t+\sigma^0_t  Z^{0,\theta,(n)}_t \right) }{(1-\gamma)(\sigma^2_t+(\sigma^0_t)^2)}\left(  \sigma_tZ^{\bar o}_t+\sigma^0_tZ^{0,\bar o}_t+h_t		\right)+ \frac{\gamma}{1-\gamma}K^{ (\theta,n-1),(\theta,1) }_t  \right.\\
		&~+\sum_{j=1}^{n-2}\frac{\gamma}{1-\gamma}K^{(j),(\theta,n-j) }_t 			
		+ Z^{\bar o}_t  Z^{\theta,(n)}_t+ Z^{0,\bar o}_t  Z^{0,\theta,(n)}_t \\ &~+\frac{1}{2}  Z^{\theta,(n-1)}_t   Z^{\theta,(1)}_t+\frac{1}{2}  Z^{0,\theta,(n-1)}_t   Z^{0,\theta,(1)}_t\\
		&~\left.+\frac{1}{2}\sum_{j=1}^{n-2}  Z^{(j)}_t  Z^{\theta,(n-j)}_t+\frac{1}{2}\sum_{j=1}^{n-2}  Z^{0,(j)}_t  Z^{0,\theta,(n-j)}_t\right\}\,dt-  Z^{\theta,(n)}_t\,dW_t-  Z^{0,\theta,(n)}_t\,dW^0_t\\
		\widehat X^{\theta,(n)}_0=&~0,~  Y^{\theta,(n)}_T=-\gamma\mathbb E[ \widehat X^{\theta,(n-1)}_T|\mathcal F^0_T	]
	\end{split}\right.
\end{equation}

%
and
\begin{equation}\label{limit-FBSDE-n}
	\left\{\begin{split}
		d\widehat X^{(n)}_t=&~\Bigg(	\frac{	\sigma_t  Z^{(n)}_t+\sigma^0_t  Z^{0,(n)}_t	}{(1-\gamma)( \sigma^2_t+(\sigma^0_t)^2 )}	\left\{	h_t-\frac{1}{1-\gamma}(	h_t+\sigma_tZ^{\bar o}_t+\sigma^0_tZ^{0,\bar o}_t	)	\right\}\\ &~-\sum_{j=1}^{n-1}\frac{1}{(1-\gamma)^2}K_t^{(j),(n-j)}\Bigg)\,dt	
		+\frac{ \left( \sigma_t  Z^{(n)}_t+\sigma^0_t  Z^{0,(n)}_t \right) }{(1-\gamma)(\sigma^2_t+(\sigma^0_t)^2)}\Big( \sigma_t\,dW_t+\sigma^0_t\,dW^0_t		\Big)   \\
		-d  Y^{(n)}_t=&~\left\{\frac{ \gamma\left( \sigma_t  Z^{(n)}_t+\sigma^0_t  Z^{0,(n)}_t \right) }{(1-\gamma)(\sigma^2_t+(\sigma^0_t)^2)}\left(  \sigma_tZ^{\bar o}_t+\sigma^0_tZ^{0,\bar o}_t+h_t		\right) +\sum_{j=1}^{n-1}\frac{\gamma}{1-\gamma}K_t^{(j),(n-j)} \right.\\
		&~\left.
		+ Z^{\bar o}_t  Z^{(n)}_t+ Z^{0,\bar o}_t  Z^{0,(n)}_t  +\frac{1}{2}\sum_{j=1}^{n-1}  Z^{(j)}_t  Z^{(n-j)}_t+\frac{1}{2}\sum_{j=1}^{n-1}  Z^{0,(j)}_t  Z^{0,(n-j)}_t\right\}\,dt \\
		&~-  Z^{(n)}_t\,dW_t-  Z^{0,(n)}_t\,dW^0_t\\
		\widehat X^{(n)}_0=&~0,~  Y^{(n)}_T=-\gamma\mathbb E[ \widehat X^{(n-1)}_T|\mathcal F^0_T	].
	\end{split}\right.
\end{equation}
In the FBSDEs \eqref{FBSDE-n} and \eqref{limit-FBSDE-n}, we use the convention that the sum vanishes whenever the lower bound of the index is larger than the upper bound of the index.

The next theorem is our main result in this section, which establishes the expansion result of the FBSDE, and thus the expansion result of the value function and the optimal investment.  
%

\begin{theorem}
Let \textbf{Assumption 1} hold. We have for any $n\geq 1$
	\begin{equation}\label{expansion-XY}
		\begin{split}
		\widehat X^\theta=&~X^{\bar o}+\sum_{i=1}^n\theta^i\widehat X^{(i)}+o(\theta^n)	,\qquad 		Y^\theta=Y^{\bar o}+\sum_{i=1}^n\theta^i  Y^{(i)}+o(\theta^n),	\\
		Z^\theta=&~Z^{\bar o}+\sum_{i=1}^n\theta^i  Z^{(i)}+o(\theta^n),\qquad Z^{0,\theta}=Z^{0,\bar o}+\sum_{i=1}^n\theta^i  Z^{0,(i)}+o(\theta^n),
		\end{split}
	\end{equation}
	where $(\widehat X^{(n)},  Y^{(n)},  Z^{(n)},  Z^{0,(n)})$ is the unique solution to \eqref{limit-FBSDE-n}. 
 In particular, 
	\[
		\log\frac{V^{\theta}}{V^{\bar o}}=\sum_{i=1}^n \theta^i\left(	\gamma\widehat X^{(i)}  +   Y^{(i)}	\right) + o(\theta^n),\qquad \pi^\theta-\pi^{\bar o}= \frac{\sigma\sum_{i=1}^n\theta^i  Z^{(i)}+\sigma^0\sum_{i=1}^n\theta^i  Z^{0,(i)}}{(1-\gamma)(\sigma^2+(\sigma^0)^2)}+o(\theta^n).
	\]
\end{theorem}
\begin{proof}
	The proof is done by induction. Lemma \ref{lem:convergence-XYZtheta-(1)} implies that
	\[
	\widehat X^\theta=X^{\bar o}+\theta\widehat X^{(1)}+o(\theta),\quad Y^\theta=Y^{\bar o}+\theta  Y^{(1)}+o(\theta),\quad Z^\theta=Z^{\bar o}+\theta Z^{(1)}+o(\theta),\quad Z^{0,\theta}=Z^{0,\bar o}+\theta  Z^{0,(1)}+o(\theta),
	\]
	which verifies \eqref{expansion-XY} for $n=1$. Furthermore, by definition, it holds that $  \varphi^{\theta,(1)}=\frac{\varphi^\theta-\varphi^{
	\bar o}}{\theta}$, $\varphi=\widehat X,Y,Z,Z^0$.
	
Now assume that the result holds for $n\geq 2$, i.e., there exists a unique tuple $(\widehat X^{\theta,(n)},  Y^{\theta,(n)},  Z^{\theta,(n)},  Z^{0,\theta,(n)})\in\bigcap\limits_{p>1}S^p\times\bigcap\limits_{p>1}S^p\times\bigcap\limits_{p>1}M^p\times\bigcap\limits_{p>1}M^p$ and a unique tuple $(\widehat X^{(n)},  Y^{(n)},  Z^{(n)},  Z^{0,(n)})\in\bigcap\limits_{p>1}S^p\times\bigcap\limits_{p>1}S^p\times\bigcap\limits_{p>1}M^p\times\bigcap\limits_{p>1}M^p$ satisfying \eqref{FBSDE-n} and \eqref{limit-FBSDE-n}, respectively, and for each $p>1$
\begin{equation}\label{eq:convergence-XYZ-theta-n}
	\lim_{\theta\rightarrow 0}\left(\|\widehat X^{\theta,(n)} - \widehat X^{(n)} \|_{S^p}+\|  Y^{\theta,(n)} -   Y^{(n)} \|_{S^p}+\|  Z^{\theta,(n)} -   Z^{(n)} \|_{M^p}+\|  Z^{0,\theta,(n)} -   Z^{0,(n)} \|_{M^p}\right)=0,
\end{equation}
and
\begin{equation}\label{eq:induction-ass}
	  \varphi^{\theta,(n)}=\frac{\varphi^\theta-\varphi^{\bar o}-\sum_{i=1}^{n-1}\theta^i  \varphi^{(i)}}{\theta^n},\qquad \varphi=\widehat X,Y,Z,Z^0.
\end{equation}

It remains to be shown that the above results also hold for $n+1$. 
Define
\[
		  \varphi^{\theta,(n+1)}=\frac{  \varphi^{\theta,(n)}-  \varphi^{(n)}}{\theta},\quad \varphi=\widehat X,Y,Z,Z^0,
\]
which implies by \eqref{eq:induction-ass} 
\[
		  \varphi^{\theta,(n+1)}=\frac{\varphi^\theta-\varphi^{\bar o}-\sum_{i=1}^{n}\theta^i   \varphi^{(i)}}{\theta^{n+1}},\qquad \varphi=\widehat X,Y,Z,Z^0.
\] 
It can also be verified directly that $(\widehat X^{\theta,(n+1)},    Y^{\theta,(n+1)},  Z^{\theta,(n+1)},  Z^{0,\theta,(n+1)})$ satisfies \eqref{FBSDE-n} with $n$ replaced by $n+1$. By the argument in the proof of Lemma \ref{lem:convergence-XYZtheta-(1)}(1), \eqref{limit-FBSDE-n} is wellposed with $n$ replaced by $n+1$. Denote by $( \widehat X^{(n+1)},    Y^{(n+1)},   Z^{(n+1)},   Z^{0,(n+1)})$ the unique solution to \eqref{limit-FBSDE-n}. By \eqref{eq:convergence-XYZ-theta-n} and the same argument in the proof of Lemma \ref{lem:convergence-XYZtheta-(1)}(2), we have for each $p>1$
\begin{equation*}\label{eq:convergence-XYZ-theta-n+1}
	\begin{split}
	\lim_{\theta\rightarrow 0}\Big(&~ \|\widehat X^{\theta,(n+1)} - \widehat X^{(n+1)} \|_{S^p}+\|  Y^{\theta,(n+1)} -   Y^{(n+1)} \|_{S^p} \\
	&~ +\|  Z^{\theta,(n+1)} -   Z^{(n+1)} \|_{M^p}+\|  Z^{0,\theta,(n+1)} -   Z^{0,(n+1)} \|_{M^p}\Big)=0.
	\end{split}
\end{equation*}
Thus, there exists $\Delta_\theta$ with $\lim_{\theta\rightarrow 0}\Delta_\theta=0$ such that
\[
	  \varphi^{\theta,(n+1)}=  \varphi^{(n+1)}+\Delta_\theta,\qquad \varphi=\widehat X,Y,Z,Z^0,
\]
which implies that by the definition of $ \varphi^{\theta,(n+1)}$
\[
	 \varphi^{\theta,(n)}= \varphi^{(n)}+\theta \varphi^{(n+1)}+\theta\Delta_\theta.
\]
By the induction assumption \eqref{eq:induction-ass}, it holds that
\[
	\varphi^\theta=\varphi^{\bar o}+\sum_{i=1}^{n-1}\theta^i \varphi^{(i)}+\theta^n \varphi^{(n)}+\theta^{n+1} \varphi^{(n+1)}+\theta^{n+1}\Delta_\theta,
\]
which implies \eqref{expansion-XY} with $n+1$.
\end{proof}

\section{ Comments on $N$-Player Games   }\label{sec:model-N-player} 
In this section, we comment on the $N$-player game \eqref{PU-intro}-\eqref{wealth-Pi-power} introduced in the Introduction. By the same argument as in Proposition \ref{lem:NE-BSDE-power}, the NE of the $N$-player game is equivalent to a multidimensional FBSDE, whose wellposedness can be obtained by the same argument in Section \ref{sec:wellposedness}, and poses no essential difference other than notational complexity. Therefore, we will omit the detailed proof for the solvability of the FBSDE with general market parameters, but instead, we discuss the case in which all market parameters are deterministic functions and connect our result with the $N$-player games studied in \cite{ET-2015,FR-2011,LZ-2019}.

By the same argument as in Proposition \ref{lem:NE-BSDE-power}, the NE $(\pi^{1,*},\cdots,\pi^{N,*})$ of the $N$-player game with power utility functions \eqref{PU-intro}-\eqref{wealth-Pi-power} is equivalent to the following multidimensional FBSDE
\begin{equation}\label{FBSDE-Np-power}
	\left\{\begin{split}
		d\widehat X^i_t=&~	\frac{	h^i_t+\sigma^i_tZ^{ii}_t+\sigma^{i0}_tZ^{i0}_t	}{(1-\gamma^i)( (\sigma^i_t)^2+(\sigma^{i0}_t)^2 )}\left\{ \left(	h^i_t-\frac{h^i_t+\sigma^i_tZ^{ii}_t+\sigma^{i0}_tZ^{i0}_t}{2(1-\gamma^i)}\right)\,dt+\sigma^i_t\,dW^i_t+\sigma^{i0}_t\,dW^0_t		\right\}, 		\\
		-dY^i_t=&~\left\{	\frac{(Z^{ii}_t)^2+(Z^{i0}_t)^2}{2}+\frac{\sum_{j\neq i}(Z^{ij}_t)^2 }{2}+\frac{\gamma^i\left(h^i_t+\sigma^i_tZ^{ii}_t+\sigma^{i0}_tZ^{i0}_t	\right)^2}{2(1-\gamma^i) ( (\sigma^i_t)^2+(\sigma^{i0}_t)^2 ) }  \right\}\,dt\\
		&~-Z^{ii}_t\,dW^i_t-Z^{i0}_t\,dW^0_t-\sum_{j\neq i}Z^{ij}_t\,dW^j_t,\\
		\widehat X^i_0=&~\log x^i_{},~Y^i_T=- \frac{\theta^i\gamma^i}{N-1}\sum_{j\neq i}\widehat X^j_T,
	\end{split}\right.
\end{equation}
with $\widehat X^i=\log X^i$ denoting the logarithm of the wealth process, and
\begin{equation}\label{optimal-strategy-i-power}
\pi^{i,*}=\frac{h^i+\sigma^i Z^{ii}+\sigma^{i0}Z^{i0}}{(1-\gamma^i)(  (\sigma^i)^2+(\sigma^{i0})^2  )}.
\end{equation}
The proof of the wellposedness for \eqref{FBSDE-Np-power} is the same as that for \eqref{FBSDE-power-2}; we compare  \eqref{FBSDE-Np-power} with the system when $\theta^i=0$ and transform the resulting FBSDE to a BSDE by the approach in Proposition \ref{prop:transformation}, followed by an {\it a priori} estimate and fixed point argument as in Lemma \ref{wellposedness-general-QBSDE} and Theorem \ref{thm:wellpos-MF-BSDE}. Instead of the detailed proof, in the next theorem, we obtain the explicit expression of the NE when market parameters are deterministic.


\begin{theorem}\label{thm:nonconstant-NE}
	When all of the coefficients $h^i$, $\sigma^i$ and $\sigma^{i0}$, $i=1,\cdots,N$, are deterministic functions of time, the unique NE is $(\pi^{1,*},\cdots,\pi^{N,*})$, where the unique bounded optimal strategy for player $i$ is
\begin{equation}\label{optimal-strategy-i-power:constant}
	\begin{split}
		\pi^{i,*}_t=&~\frac{h^i_t  }{  (1-\gamma^i)(\sigma^i_t)^2+\left(	1-\gamma^i-\frac{\theta^i\gamma^i}{N-1}	\right)(\sigma^{i0}_t)^2 }	\\
		&~- \frac{\theta^i\gamma^i\sigma^{i0}_t}{   (1-\gamma^i)(\sigma^i_t)^2+\left(	1-\gamma^i-\frac{\theta^i\gamma^i}{N-1}	\right)(\sigma^{i0}_t)^2 }\frac{\varphi^{(1),N}_t}{1+\varphi^{(2),N}_t}, \quad t\in[0,T],\quad i=1,\cdots,N,
	\end{split}
\end{equation}
where
\begin{equation*}
	\begin{split}
		\varphi^{(1),N}=\frac{1}{N-1}\sum_{j=1}^N  \frac{h^j\sigma^{j0}}{ (1-\gamma^j)(\sigma^j)^2 +\left(1-\gamma^j-\frac{\theta^j\gamma^j}{N-1}	\right)(\sigma^{j0})^2	 }
	\end{split}
\end{equation*}
and
\[
	\varphi^{(2),N}=\frac{1}{N-1}\sum_{j=1}^N \frac{\theta^j\gamma^j(\sigma^{j0})^2}{			(1-\gamma^j)(\sigma^j)^2	+\left(	1-\gamma^j-\frac{\theta^j\gamma^j}{N-1}	\right)(\sigma^{j0})^2		}.
\]
In addition, if $\{(x^i,\theta^i,\gamma^i, h^i,\sigma^i,\sigma^{i0})\}_{i=1}^N$ are realizations from independent copies of $(x,\theta, \gamma,  h,\sigma,\sigma^0)$ in Theorem \ref{thm:constant-equilibrium-MFG}, then \eqref{optimal-strategy-i-power:constant} converges to \eqref{eq:non-constant-strategy-power}. 
\end{theorem}
\begin{proof}
One can show that \eqref{FBSDE-Np-power} admits at most one solution in $(S^2\times S^2\times L^\infty\times L^\infty\times L^\infty)^N$.	The unique bounded optimal strategy is constructed as follows.
Let $\overline{\widehat X}^{-i}=\frac{1}{N-1}\sum_{j\neq i}\widehat X^j$ and $\overline{\widehat x}^{-i}=\frac{1}{N-1}\sum_{j\neq i}\log(x^j)$. 
By taking the average of the forward dynamics of \eqref{FBSDE-Np-power}, we get $\overline{\widehat X}^{-i}$, and by taking $\overline{\widehat X}^{-i}$ into the backward dynamics of \eqref{FBSDE-Np-power}, it holds that
\begin{align*}
		Y^i_t=&~-\theta^i\gamma^i\overline{\widehat x}^{-i}-\frac{\theta^i\gamma^i}{N-1}\sum_{j\neq i}\int_0^T\frac{	h^j_s+\sigma^j_sZ^{jj}_s+\sigma^{j0}_sZ^{j0}_s	}{(1-\gamma^j)( (\sigma^j_s)^2+(\sigma^{j0}_s)^2 )} \left(	h^j_s-\frac{h^j_s+\sigma^j_sZ^{jj}_s+\sigma^{j0}_sZ^{j0}_s}{2(1-\gamma^j)}\right)\,ds	\\
		&~+\int_t^T\left\{	\frac{(Z^{ii}_s)^2+(Z^{i0}_s)^2}{2}+\frac{\sum_{j\neq i}(Z^{ij}_s)^2 }{2}+\frac{\gamma^i\left(h^i_s+\sigma^i_sZ^{ii}_s+\sigma^{i0}_sZ^{i0}_s	\right)^2}{2(1-\gamma^i) ( (\sigma^i_s)^2+(\sigma^{i0}_s)^2 ) }  \right\}\,ds\\
		&~-\frac{\theta^i\gamma^i}{N-1}\sum_{j\neq i}\int_0^t\frac{\sigma^j_sZ^{jj}_s+\sigma^{j0}_sZ^{j0}_s+h^j_s		}{(1-\gamma^j)(	(\sigma^j_s)^2+(\sigma^{j0}_s)^2	)}\sigma^j_s\,dW^j_s\\
		&~-\frac{\theta^i\gamma^i}{N-1}\sum_{j\neq i}\int_0^t\frac{\sigma^j_sZ^{jj}_s+\sigma^{j0}_sZ^{j0}_s+h^j_s		}{(1-\gamma^j)(	(\sigma^j_s)^2+(\sigma^{j0}_s)^2	)}\sigma^{j0}_s\,dW^0_s\\
		&~-\int_t^TZ^{ii}_s\,dW^i_s+\int_t^T\sum_{j\neq i}\left\{ - \frac{\theta^i\gamma^i}{N-1}\frac{\sigma^j_sZ^{jj}_s+\sigma^{j0}_sZ^{j0}_s+h^j_s		}{(1-\gamma^j)(	(\sigma^j_s)^2+(\sigma^{j0}_s)^2	)}\sigma^j_s -Z^{ij}_s\right\}\,dW^j_s\\
		&~+\int_t^T\left\{-\frac{\theta^i\gamma^i}{N-1}\sum_{j\neq i}\frac{\sigma^j_sZ^{jj}_s+\sigma^{j0}_sZ^{j0}_s+h^j_s		}{(1-\gamma^j)(	(\sigma^j_s)^2+(\sigma^{j0}_s)^2	)}\sigma^{j0}_s -Z^{i0}_s\right\}  \,dW^0_s.
\end{align*}
To construct an adapted $Y^i$, let 
\begin{equation}\label{Z-i=0}
	\left\{\begin{split}
		0=&~Z^{ii}\\
		0=&~- \frac{\theta^i\gamma^i}{N-1}\frac{\sigma^jZ^{jj}+\sigma^{j0}Z^{j0}+h^j	}{(1-\gamma^j)(	(\sigma^j)^2+(\sigma^{j0})^2	)}\sigma^j -Z^{ij}\\
		0=&~-\frac{\theta^i\gamma^i}{N-1}\sum_{j\neq i}\frac{\sigma^jZ^{jj}+\sigma^{j0}Z^{j0}+h^j		}{(1-\gamma^j)(	(\sigma^j)^2+(\sigma^{j0})^2	)}\sigma^{j0} -Z^{i0}.
	\end{split}\right.
\end{equation}
The first and the third equalities in \eqref{Z-i=0} yield that
\begin{equation}\label{Z-i=0-2}
	\begin{split}
	&~\left(1-\frac{\theta^i\gamma^i}{N-1} \frac{(\sigma^{i0})^2}{(1-\gamma^i)((\sigma^i)^2+(\sigma^{i0})^2)} \right)Z^{i0}\\
	=&~-\frac{\theta^i\gamma^i}{N-1}\sum_{j=1}^N\frac{(\sigma^{j0})^2Z^{j0} }{(1-\gamma^j)((\sigma^j)^2+(\sigma^{j0})^2)}-\frac{\theta^i\gamma^i}{N-1}\sum_{j\neq i}\frac{h^j\sigma^{j0}}{(1-\gamma^j)(	(\sigma^j)^2+(\sigma^{j0})^2	)},
	\end{split}
\end{equation}
which further implies by multiplying $\frac{	\frac{(\sigma^{i0})^2}{ (1-\gamma^i)( (\sigma^{i})^2+(\sigma^{i0})^2) }	}{	1-\frac{\theta^i\gamma^i}{N-1} \frac{(\sigma^{i0})^2}{(1-\gamma^i)((\sigma^i)^2+(\sigma^{i0})^2)}	}$ and taking the sum from $1$ to $N$ on both sides
\begin{equation*}
	\begin{split}
&~\sum_{i=1}^N\frac{(\sigma^{i0})^2Z^{i0}}{ (1-\gamma^i) ((\sigma^{i})^2+(\sigma^{i0})^2) }\\
=&~\sum_{i=1}^N  \frac{ - \frac{\theta^i\gamma^i}{N-1} \frac{(\sigma^{i0})^2}{(1-\gamma^i)((\sigma^i)^2+(\sigma^{i0})^2)} }{   1-\frac{\theta^i\gamma^i}{N-1}\frac{(\sigma^{i0})^2}{(1-\gamma^i)((\sigma^i)^2+(\sigma^{i0})^2)}   } \sum_{j=1}^N\frac{(\sigma^{j0})^2Z^{j0}}{ (1-\gamma^j) ((\sigma^{j})^2+(\sigma^{j0})^2) }\\
&~+\sum_{i=1}^N \frac{ - \frac{\theta^i\gamma^i}{N-1} \frac{(\sigma^{i0})^2}{(1-\gamma^i)((\sigma^i)^2+(\sigma^{i0})^2)}   }{   1-\frac{\theta^i\gamma^i}{N-1}\frac{(\sigma^{i0})^2}{(1-\gamma^i)((\sigma^i)^2+(\sigma^{i0})^2)}   }\sum_{j\neq i} \frac{h^j\sigma^{j0}}{(1-\gamma^j)(	(\sigma^j)^2+(\sigma^{j0})^2	)} .
	\end{split}
\end{equation*}
From the above linear equation for $\sum_{i=1}^N\frac{(\sigma^{i0})^2Z^{i0}}{ (1-\gamma^i) ((\sigma^{i})^2+(\sigma^{i0})^2) }$, we get
\begin{equation*}
	\begin{split}
	&~\sum_{i=1}^N\frac{(\sigma^{i0})^2Z^{i0}}{ (1-\gamma^i)( (\sigma^{i})^2+(\sigma^{i0})^2) }\\
	=&~\frac{-1}{1+   \sum_{i=1}^N  \frac{  \frac{\theta^i\gamma^i}{N-1} \frac{(\sigma^{i0})^2}{(1-\gamma^i)((\sigma^i)^2+(\sigma^{i0})^2)}   }{   1-\frac{\theta^i\gamma^i}{N-1}\frac{(\sigma^{i0})^2}{(1-\gamma^i)((\sigma^i)^2+(\sigma^{i0})^2)}   } }	\sum_{i=1}^N  \frac{ \frac{\theta^i\gamma^i}{N-1} \frac{(\sigma^{i0})^2}{(1-\gamma^i)((\sigma^i)^2+(\sigma^{i0})^2)}   }{   1-\frac{\theta^i\gamma^i}{N-1}\frac{(\sigma^{i0})^2}{(1-\gamma^i)((\sigma^i)^2+(\sigma^{i0})^2)}   }	\sum_{j\neq i} \frac{h^j\sigma^{j0}}{(1-\gamma^j)(	(\sigma^j)^2+(\sigma^{j0})^2	)}    .
	\end{split}
\end{equation*}
Taking the equality back into \eqref{Z-i=0-2}, we obtain
\begin{equation*}
	\begin{split}
		Z^{i0}=&~-\frac{ \frac{\theta^i\gamma^i}{N-1} }{  1-\frac{\theta^i\gamma^i(\sigma^{i0})^2}{(N-1)(1-\gamma^i) \{ (\sigma^{i0})^2+(\sigma^0)^2 \}}  }\sum_{j\neq i}\frac{ \sigma^{j0}h^j }{(1-\gamma^j)	\{( \sigma^j )^2+( \sigma^{j0} )^2\}	}\\
		&~+\frac{	\frac{\theta^i\gamma^i}{N-1}	}{	1-\frac{\theta^i\gamma^i(\sigma^{i0})^2}{ (N-1)(1-\gamma^i)\{  (\sigma^i)^2+(\sigma^{i0})^2  \}  }	} \frac{1}{1+\varphi^{(2),N}}\sum_{i=1}^N  \frac{ \frac{\theta^i\gamma^i(\sigma^{i0})^2}{(N-1)(1-\gamma^i)\{ (\sigma^i)^2+( \sigma^{i0} )^2 \}}  }{1-\frac{\theta^i\gamma^i(\sigma^{i0})^2}{(N-1)(1-\gamma^i)\{ (\sigma^i)^2+(\sigma^{i0})^2 \}}}\sum_{j\neq i} \frac{\sigma^{j0}h^j}{(1-\gamma^j) \{  (\sigma^j)^2+(\sigma^{j0})^2 \} },
	\end{split}
\end{equation*}
where $\varphi^{(2),N}$ is defined in the statement of the theorem.
Consequently, the optimal strategy \eqref{optimal-strategy-i-power:constant} can be obtained from \eqref{optimal-strategy-i-power}. The convergence from \eqref{optimal-strategy-i-power:constant} to \eqref{eq:non-constant-strategy-power} is obtained by the law of large numbers.
%
\end{proof}

\begin{remark}
This remark discusses the link between our $N$-player game and the	$N$-player games studied by Espinosa and Touzi in \cite{ET-2015}, where all players trade common stocks; in our model, the stock prices are driven by both idiosyncratic noise and common noise.
First, when there is no trading constraint, Espinosa and Touzi obtained a unique NE by convex duality for general utility functions; refer to \cite[Theorem 3.3]{ET-2015}. Note that the analysis in \cite{ET-2015} can cover the case of power utility functions as in our paper, although power utility functions do not satisfy the Inada conditions in \cite[(2.4)]{ET-2015}. Second, when there is a general trading constraint, Espinosa and Touzi obtained a unique NE for exponential utility functions by assuming the market parameters to be deterministic and continuous; refer to \cite[Theorem 4.8]{ET-2015}. This NE was called \textit{deterministic Nash equilibrium,} which was also studied in \cite{FR-2011}. When there is a linear trading constraint, \cite[Theorem 5.2]{ET-2015} obtained an NE in closed form for exponential utility functions and deterministic market parameters. Thus, our paper partially recovers \cite[Theorem 5.2]{ET-2015} given Remark \ref{remark:exponential}. 
Third, \cite[Example 5.12--5.16]{ET-2015} obtained a closed form NE under various trading constraints by assuming the price processes of risky assets to be independent. The $N$-player game with independence assumption is similar to our model when the individual volatility $\sigma^i=0$ for all $i$; however, we have no trading constraint.
Fourth, in our paper, both idiosyncratic noise and common noise are one-dimensional, so we cannot study \textit{correlated investments} as in \cite[Example 5.17]{ET-2015}, although generalization to multi-dimension can be expected. Finally, \cite{ET-2015} also investigated convergence from $N$-player games to MFGs under deterministic market parameters; refer to \cite[Proposition 4.12, Proposition 5.7, Example 5.8 and Example 5.9]{ET-2015}.
\end{remark}

\begin{remark}[Comments on  closed-loop NE]\label{rmk:closed-lopp}     The closed-loop NE may not exist. However, once it exists, the closed-loop NE must be identical to the open-loop NE. Indeed, 
	If we assume all competing players except player $i$ use a closed-loop strategy $\pi^j(t,\theta,\gamma,\bm X)$ with $\bm X=(X^1,\cdots,X^N)$, then the same argument leading to \eqref{FBSDE-Np-power} implies that $\pi^{i,*}$ has the same expression as  \eqref{optimal-strategy-i-power}. If the market parameters $(h^i,\sigma^i,\sigma^{i0})$ are progressively measurable with respect to the filtration generated by the Browian motions, $\{Z^{ij}\}_{j=0,1,\cdots,N}$ are not necessarily deterministic functions of $\bm X$ since the nonsingularity of $(\pi^{i,*})^2(	(\sigma^i)^2+(\sigma^{i0})^2	)$ cannot be guaranteed. Consequently, the closed-loop equilibrium may not exist. 
	 Furthermore, under the assumptions in Theorem \ref{thm:nonconstant-NE}, the open-loop NE \eqref{optimal-strategy-i-power:constant}  is also a closed-loop one. 
	 
	 This comment also applies to the MFG \eqref{model-MFG-power}.
\end{remark}
As a corollary of Theorem \ref{thm:nonconstant-NE}, we recover the $N$-player games in \cite{LZ-2019} and conclude that the constant equilibrium is unique in $L^\infty$.
\begin{corollary}[\textit{Lacker \& Zariphopoulou's $N$-player games revisited}]
Assume that all of the coefficients $h^i$, $\sigma^i$ and $\sigma^{i0}$ are constants, then the NE obtained in  \eqref{optimal-strategy-i-power:constant} is unique in $L^\infty$.
Furthermore,  \eqref{optimal-strategy-i-power:constant} is consistent with \cite[Theorem 14]{LZ-2019} by taking \cite[Remark 16]{LZ-2019} into account. 
\end{corollary} ~\\[-20mm]

\section{Conclusion} ~\\[-8mm] 
In this paper we study mean field portfolio games with random market parameters. We establish a one-to-one correspondence between the NE of the portfolio game and the solution to some mean field FBSDE. The unique NE is obtained by solving the FBSDE under a weak interaction assumption. When the market parameters do not depend on the Brownian paths, we get the NE in closed form. Our result partially generalizes the results in \cite{ET-2015} and completely generalizes the results in \cite{LZ-2019}. Moreover, motivated by the weak interaction assumption, we establish an asymptotic expansion result in powers of the competition parameter $\theta$. In particular, we expand the log-value function and the optimal investment in powers of $\theta$ into any order. This result allows us to obtain the value function and the optimal investment based only on the benchmark model without competition.  Although our paper focuses on the case of power utility functions, our analysis can also be extended to cases of exponential and log utility functions.
\begin{appendix}

\section{$\theta$-Dependent Terms in the BSDE \eqref{MF-BSDE}}\label{sec:J2}
In this section, we summarize the cumbersome $\theta$-dependent terms in the BSDE \eqref{MF-BSDE}.  To facilitate the presentation, we introduce the following notation
\begin{align*}\label{eq:notation}
	&~\phi^\sigma=1+\frac{\gamma\sigma^2}{(1-\gamma)(\sigma^2+(\sigma^0)^2)},\quad \phi^{\sigma^0}=1+\frac{\gamma(\sigma^0)^2}{(1-\gamma)(\sigma^2+(\sigma^0)^2)},\quad 	g_\cdot=-\mathbb E\left[\left.  \frac{\theta\gamma(\sigma^0_\cdot)^2	}{(1-\gamma)(  \sigma^2_{\cdot}+(\sigma^0_{\cdot})^2  )	}  \right|\mathcal F^0_\cdot		\right],\\
		&~	f^\sigma=\frac{\sigma h}{(1-\gamma)(\sigma^2+(\sigma^0)^2 )},\quad 	f^{\sigma^0}=\frac{\sigma^0 h}{(1-\gamma)(\sigma^2+(\sigma^0)^2 )},\quad f^{h}=\frac{ h^2}{(1-\gamma)(\sigma^2+(\sigma^0)^2 )},\\
		&~\psi=\frac{ \sigma\sigma^0}{ (1-\gamma)( \sigma^2+(\sigma^0)^2  ) },\quad \psi^\sigma=\frac{\sigma^2}{ (1-\gamma)( \sigma^2+(\sigma^0)^2  ) },\quad \psi^{\sigma^0}=\frac{(\sigma^0)^2}{ (1-\gamma)( \sigma^2+(\sigma^0)^2  ) }, \\
		&~\phi^{(1)}=	\frac{h\sigma^0}{(1-\gamma)(	\sigma^2+(\sigma^0)^2	)}+	\frac{\sigma\sigma^0Z^{\bar o}}{(1-\gamma)(	\sigma^2+(\sigma^0)^2	)}	+	\frac{(\sigma^0)^2Z^{0,\bar o}}{(1-\gamma)(\sigma^2+(\sigma^0)^2)},\\
		&~\phi^{(2)}=	\frac{h\sigma}{(1-\gamma)(	\sigma^2+(\sigma^0)^2	)}+	\frac{\sigma^2Z^{\bar o}}{(1-\gamma)(	\sigma^2+(\sigma^0)^2	)}	+	\frac{\sigma\sigma^0Z^{0,\bar o}}{(1-\gamma)(\sigma^2+(\sigma^0)^2)},\\
		&~\phi^{(3)}=Z^{0,\bar o}  + \frac{\gamma\sigma^0 (h+\sigma Z^{\bar o}+\sigma^0Z^{0,\bar o}	) }{(1-\gamma)(\sigma^2+(\sigma^0)^2)},\\
		&~\phi^{(4)}= Z^{\bar o}+\frac{\gamma\sigma( h+\sigma Z^{\bar o}+\sigma^0Z^{0,\bar o} ) }{(1-\gamma)(\sigma^2+(\sigma^0)^2)}.
\end{align*}
The terms that are dependent on $\theta$ follow
\[
	\mathcal J_2(t;\widetilde Z,\widetilde Z^0,\theta)=\mathcal I_1(t;\widetilde Z)+\mathcal I_2(t;\widetilde Z^0)+\mathcal I_3(t;\widetilde Z,\widetilde Z^0)+\mathcal I_4(t),
\]
 where the terms involving $\widetilde Z$ are given by
\begin{equation*}
	\begin{split}
		\mathcal I_1(t;\widetilde Z)
		=&~\left\{ \frac{\phi^{\sigma^0}_t\theta^2\gamma^2}{2(1-g_t)^2}+\theta\gamma\mathbb E\left[\left.	\frac{ \theta^2\gamma^2 \psi^{\sigma^0}_t	 }{2(1-\gamma)(1-g_t)^2}\right|\mathcal F^0_t	\right]\right\}\left(	\mathbb E[\psi_t\widetilde Z_t|\mathcal F^0_t]	\right)^2+\theta\gamma\mathbb E\left[\left. \frac{\psi^\sigma_t}{2(1-\gamma)}	(\widetilde Z_t)^2	\right|\mathcal F^0\right]\\
		&~-\frac{\theta\gamma^2\psi_t}{1-g_t}\widetilde Z_t\mathbb E[\psi_t\widetilde Z_t|\mathcal F^0_t]-\theta\gamma \mathbb E\left[	\left.	\frac{\theta\gamma\psi_t}{(1-\gamma)(1-g_t)}\widetilde Z_t	\right|\mathcal F^0_t	\right]\mathbb E[ \psi_t\widetilde Z_t|\mathcal F^0_t ]-\theta\gamma\mathbb E[	f^\sigma_t\widetilde Z_t|\mathcal F^0_t	]\\
		&~+\left\{	\frac{\theta\gamma}{1-g_t}\mathbb E[	\theta\gamma f^{\sigma^0}_t|\mathcal F^0_t	]	-\frac{\theta\gamma}{1-g_t}\phi^{(3)}_t +\frac{\theta^2\gamma^2}{(1-g_t)^2}\phi^{\sigma^0}_t\mathbb E[	\phi^{(1)}_t|\mathcal F^0_t	]	\right\} \mathbb E[ \psi_t\widetilde Z_t|\mathcal F^0_t ]\\
		&~+\theta\gamma\left\{ - \mathbb E\left[ \left.	\frac{\theta\gamma\phi_t^{(1)}}{(1-\gamma)(1-g_t)}\right|\mathcal F^0_t\right]  + \mathbb E\left[ \left.	\frac{\theta^2\gamma^2\psi_t^{\sigma^0}}{(1-\gamma)(1-g_t)^2}	\right|\mathcal F^0_t\right]	\mathbb E[ \phi^{(1)}_t |\mathcal F^0_t ]		\right\}\mathbb E[ \psi_t\widetilde Z_t|\mathcal F^0_t ]\\
		&~+\theta\gamma\mathbb E\left[	\left.	 \frac{\phi_t^{(2)}}{1-\gamma}\widetilde Z_t	\right|\mathcal F^0_t	\right]-\theta\gamma \mathbb E[ \phi^{(1)}_t |\mathcal F^0_t] \mathbb E\left[	\left.		\frac{ \theta  \gamma\psi_t}{(1-\gamma)(1-g_t)} \widetilde Z_t	\right|\mathcal F^0_t	\right]+\left\{\phi_t^{(4)}-\frac{\theta\gamma^2}{1-g_t}\psi_t\mathbb E[\phi^{(1)}_t|\mathcal F^0_t]	\right\} \widetilde Z_t,
	\end{split}
\end{equation*}	

the terms involving $\widetilde Z^0$ are given by
\begin{align*}
	\mathcal I_2(t;\widetilde Z^0)
	=&~\left\{\frac{\phi_t^{\sigma^0}\theta^2\gamma^2}{2(1-g_t)^2}+\theta\gamma		\mathbb E\left[	\left.	\frac{\theta^2  \gamma^2\psi_t^{\sigma^0}}{2(1-\gamma)(1-g_t)^2}		\right|\mathcal F^0_t	\right] \right\}\left(	\mathbb E[\psi_t^{\sigma^0}\widetilde Z^0_t|\mathcal F^0_t	]	\right)^2 -\frac{\theta\gamma\phi^{\sigma^0}_t}{1-g_t}\widetilde Z^0_t\mathbb E[ \psi^{\sigma^0}_t\widetilde Z^0_t|\mathcal F^0_t ]\\
	&~+\theta\gamma  	\mathbb E\left[	\left.	\frac{  \psi^{\sigma^0}_t}{2(1-\gamma)}	(\widetilde Z^0_t)^2	\right|\mathcal F^0	\right]	 -\theta\gamma	\mathbb E\left[	\left.	\frac{\theta  \gamma\psi_t^{\sigma^0}}{(1-\gamma)(1-g_t)}	\widetilde Z^0_t	\right|\mathcal F^0_t	\right]\mathbb E[\psi_t^{\sigma^0}\widetilde Z^0_t|\mathcal F^0_t	]	\\
	&~+	\frac{\theta\gamma}{1-g_t}\left\{\mathbb E[	\theta\gamma f^{\sigma^0}_t|\mathcal F^0_t	]	-\phi^{(3)}_t+\frac{\theta\gamma}{1-g_t}\phi^{\sigma^0}_t\mathbb E[ \phi^{(1)}_t|\mathcal F^0_t ]	\right\}\mathbb E[ \psi^{\sigma^0}_t\widetilde Z^0_t|\mathcal F^0_t ]\\
	&~	+\theta\gamma \left\{  - \mathbb E\left[\left.    \frac{\theta\gamma \phi^{(1)}_t }{(1-\gamma)(1-g_t)}   \right|\mathcal F^0_t			\right] +\mathbb E\left[	\left.  \frac{ \theta^2  \gamma^2\psi^{\sigma^0}_t }{(1-\gamma)(1-g_t)^2}		\right|\mathcal F^0_t\right]\mathbb E[\phi^{(1)}_t|\mathcal F^0_t] \right\}\mathbb E[ \psi^{\sigma^0}_t \widetilde Z^0_t|\mathcal F^0_t ] \\
	&~+\theta\gamma \mathbb E\left[\left.	 \frac{  \phi_t^{(1)}}{1-\gamma} \widetilde Z^0_t	 \right|\mathcal F^0_t	\right]		-\theta\gamma \mathbb E[\phi^{(1)}_t|\mathcal F^0_t] \mathbb E\left[  \left.	 \frac{\theta\gamma \psi^{\sigma^0}_t}{(1-\gamma)(1-g_t)}\widetilde Z^0_t	\right|\mathcal F^0_t	\right]-\theta\gamma\mathbb E[	f^{\sigma^0}_t\widetilde Z^0_t|\mathcal F^0_t	]\\	&~+\left\{\phi^{(3)}_t-\frac{\theta\gamma}{1-g_t}\phi^{\sigma^0}_t\mathbb E[\phi^{(1)}_t|\mathcal F^0_t]\right\}	\widetilde Z^0_t,
\end{align*}
the crossing terms involving both $\widetilde Z$ and $\widetilde Z^0$ are given by
\begin{equation*}
	\begin{split}
		&~\mathcal I_3(t;\widetilde Z,\widetilde Z^0_t)\\
		=&~	-\frac{\theta\gamma^2\psi}{1-g_t}\widetilde Z_t\mathbb E[ \psi^{\sigma^0}_t\widetilde Z^0_t|\mathcal F^0_t ]-\frac{\theta\gamma\phi_t^{\sigma^0}}{1-g_t}\mathbb E[	\psi_t\widetilde Z_t|\mathcal F^0_t	]\widetilde Z^0_t+\frac{\theta^2\gamma^2\phi_t^{\sigma^0}}{(1-g_t)^2}\mathbb E[\psi_t\widetilde Z_t|\mathcal F^0_t]\mathbb E[ \psi_t^{\sigma^0}\widetilde Z^0_t|\mathcal F^0_t ]\\
		&~	-\theta\gamma\mathbb E\left[\left.	   \frac{  \theta\gamma \psi_t^{\sigma^0}}{(1-\gamma)(1-g_t)} \widetilde Z^0_t \right|\mathcal F^0_t	\right] \mathbb E[\psi_t\widetilde Z_t|\mathcal F^0_t]+\theta\gamma \mathbb E\left[	\left. \frac{   \theta^2\gamma^2\psi_t^{\sigma^0}   }{(1-\gamma)(1-g_t)^2}  \right|\mathcal F^0_t	\right] \mathbb E[ \psi_t\widetilde Z_t|\mathcal F^0_t  ] \mathbb E[  \psi_t^{\sigma^0}\widetilde Z^0_t|\mathcal F^0_t  ]\\
		&~+\theta\gamma\mathbb E\left[	\left.   \frac{  \psi_t}{1-\gamma} \widetilde Z_t\widetilde Z^0_t \right|\mathcal F^0_t	\right]-\theta\gamma\mathbb E\left[\left. \frac{ \theta\gamma\psi_t  }{(1-\gamma)(1-g_t)}\widetilde Z_t \right|\mathcal F^0_t	\right] \mathbb E\left[	\psi_t^{\sigma^0}\widetilde Z^0_t|\mathcal F^0_t		\right],
	\end{split}
\end{equation*}
and the remaining terms are given by
\begin{equation*}
	\begin{split}
		\mathcal I_4(t)=&~	\frac{ (\theta\gamma)^2\phi^{\sigma^0}_t }{2(1-g_t)^2}(\mathbb E[\phi^{(1)}_t|\mathcal F^0_t])^2 -\frac{\theta\gamma}{1-g_t}\left(-	\mathbb E[ \theta\gamma f^{\sigma^0}_t |\mathcal F^0_t]+\phi^{(3)}_t		\right)\mathbb E[\phi^{(1)}_t|\mathcal F^0_t]\\
		&~+\theta\gamma\mathbb E\left[ \left.   \frac{(\sigma^2_t+(\sigma^0_t)^2)| \phi^{(1)}_t |^2  }{2}   \right|\mathcal F^0_t \right]-\theta\gamma\mathbb E[  \phi^{(1)}_t|\mathcal F^0_t ] \mathbb E\left[\left.	\frac{\theta\gamma\phi_t^{(1)}}{(1-\gamma)(1-g_t)}	\right|\mathcal F^0_t	\right]\\
		&~+\frac{\theta^3\gamma^3}{(1-g_t)^2}\mathbb E\left[ \left. \frac{\psi^{\sigma^0}_t}{2(1-\gamma)}  \right|\mathcal F^0_t \right]\left(	\mathbb E[ \phi^{(1)}_t|\mathcal F^0_t ]		\right)^2\\
		&~-\theta\gamma\mathbb E[ f^h_t+f_t^\sigma Z^{\bar o}_t  + f^{\sigma^0}_tZ^{0,\bar o}_t |\mathcal F^0_t].
	\end{split}
\end{equation*}

\section{Reverse H\"older's Inequality}\label{app:reverse}
We summarize the reverse H\"older's inequality for a general stochastic process and for the stochastic exponential of a stochastic process in the BMO space (\cite[Theorem 3.1]{Kazamaki-2006}), which is used in the main text. 

For some $p>1$, we say that a stochastic process $D$ satisfies $R_p$ if there exists a constant $C$, such that for any $[0,T]$-valued stopping time $\tau$, it holds that
\[
	\mathbb E\left[	\left|\frac{D_T}{D_\tau}		\right|^p\Big|\mathcal F_\tau	\right]\leq C.
\]

 Let $\Theta\in H^2_{BMO}$ and $B$ be a Brownian motion. Define  
\[
	\mathcal E_t(\Theta)=\mathcal E\left( \int_0^t\Theta_s\,dB_s   \right).
\]
Let $\Phi$ be a function defined on $(1,\infty)$: 
\[
	\Phi(x)=\left\{  1+\frac{1}{x^2}\log\frac{2x-1}{2(x-1)}			 \right\}^{\frac{1}{2}}-1,
\]
which is a continuous decreasing function satisfying $\lim_{x\searrow 1}\Phi(x)=\infty$ and $\lim_{x\nearrow\infty}\Phi(x)=0$. Let $p_{\Theta}$ be the unique constant, such that $\Phi(p_{\Theta})=\|\Theta\|_{BMO}$. Then we have the following reverse H\"older's inequality. 
\begin{lemma}\label{lem:reverse}
	For any $1<p<p_\Theta$ and any stopping time $\tau$, it holds that $\mathcal E(\Theta)$ satisfies $R_p$. In particular, 
	\[
		\mathbb E\left[\left. \left| \frac{\mathcal E_T(\Theta)}{\mathcal E_\tau(\Theta)}	\right|^p \right|\mathcal F_\tau\right]\leq K(p,\|\Theta\|_{BMO}), 
	\]
	where
	\begin{equation*}
		K(p, N)=\frac{2}{1-\frac{2(p-1)}{2p-1}e^{p^2( N^2+2N )}   }. \\[-3mm]
	\end{equation*}
\end{lemma}
\end{appendix}

\bibliography{Fu}

\end{document}